\documentclass[11pt,a4paper]{article}

\usepackage[centertags]{amsmath}
\usepackage{amsfonts}
\usepackage{amssymb}
\usepackage{amsthm}
\usepackage{epsfig}
\usepackage{setspace}
\usepackage{ae} 
\usepackage{eucal} 
\usepackage[authoryear]{natbib}

\usepackage[usenames]{color}
   
   \csname @addtoreset\endcsname{equation}{section}
\newcommand{\cb}[1]{\color{blue}#1\color{black}}
\newcommand{\nd}[1]{\color{red}#1\color{black}}
\renewcommand{\cb}[1]{#1}
\renewcommand{\nd}[1]{#1}

\theoremstyle{plain}
\newtheorem{thm}{Theorem}[section]
\newtheorem{lem}[thm]{Lemma}
\newtheorem{cor}[thm]{Corollary}
\newtheorem{prop}[thm]{Proposition}
\theoremstyle{definition}

\theoremstyle{remark}
\newtheorem{exmp}[thm]{Example}
\newtheorem{rem}[thm]{Remark}

 \DeclareMathOperator*{\esssup}{esssup}
\DeclareMathOperator*{\essinf}{essinf}
\def\upper{\nd{U}}%
\date{Submitted: February 26, 2014. Revised: November 20, 2014}

\begin{document}

\title{A First-Order BSPDE for Swing Option Pricing: \\  Classical Solutions}

\author{Christian Bender$^1$, Nikolai Dokuchaev$^2$}

 \maketitle
\footnotetext[1]{Saarland University, Department of Mathematics,
Postfach 151150, D-66041 Saarbr\"ucken, Germany, {\tt
bender@math.uni-sb.de}.} \footnotetext[2]{       Department of
Mathematics \& Statistics, Curtin University, GPO Box U1987, Perth,
6845 Western Australia, Australia, {\tt N.Dokuchaev@curtin.edu.au }}

\begin{abstract}

In \citet{BD} we studied a control problem related to swing option pricing in a general non-Markovian setting. The main result there shows that
the value process of this control problem can be uniquely characterized in terms of a first order backward SPDE and a pathwise differential inclusion.
In the present paper we additionally assume that the cashflow process of the swing option is left-continuous in expectation (LCE). Under this assumption we
show that the value process is continuously differentiable in the space variable that represents the volume which the holder of the option
can still exercise until maturity. This gives rise to an existence and uniqueness result for
the corresponding backward SPDE in a classical sense. We also explicitly represent the space derivative
of the value process in terms of a nonstandard optimal stopping problem over a subset of predictable stopping times. This representation can be applied to derive a dual
minimization problem in terms of martingales.
\\[0.2cm] \emph{Keywords:} Backward SPDE, optimal stopping, stochastic optimal control, swing options.
\\[0.2cm] \emph{AMS classification:} 60H15; 49L20; 91G20.
\end{abstract}

\section{Introduction}

Motivated by the pricing problem for swing options, we consider the following optimal control problem.
The investor's aim is to maximize the expected reward of exercising an adapted cashflow process $X$, i.e. she wishes to maximize
\begin{equation}\label{optimal_static}
E\left[\int_0^T u(s)X(s)ds\right]
\end{equation}
over all adapted processes $u$ with values in $[0,L]$ which satisfy the \cb{condition }$\int_0^T u(s)ds \leq 1$. Here $L$ is a local constraint,
which restricts the maximal rate at which the holder of the option can exercise the cashflow process $X$. Moreover, the global constraint (which is a finite fuel constraint) imposes
that the total volume spent by the holder is bounded by one. We refer to \citet{Ke04} \cb{for } the modelling of swing options as continuous time optimal control problems, and note
that the above problem and related problems were recently investigated in a Markovian diffusion setting by \citet{BLM,Do2013,BCV}.

In our companion paper \citep{BD}, to which we also refer for  further references on swing option pricing, we studied the above optimal control problem in a general non-Markovian setting
under the following mild assumptions: $(X(t),\; 0\leq t \leq T)$ is a nonnegative,
right-continuous, $\mathbb{F}$-adapted stochastic process on a filtered probability space
$(\Omega, \mathcal{F},
\mathbb{F}, P)$  satisfying the usual
conditions such that
\begin{equation}\label{ass:square}
 E[\sup_{0\leq t \leq T} X(t)^p]<\infty
\end{equation}
for some $p>1$. We consider these conditions  as standing assumptions for the rest of the paper. Then, a dynamic formulation of (\ref{optimal_static})
reads as follows:
For any $[0,T]$-valued stopping
time $\tau$ and $\mathcal{F}_\tau$-measurable, $(-\infty,1]$-valued
random variable $Y$ denote by $U(\tau,Y)$ the set of all
$\mathbb{F}$-adapted processes with values in $[0,L]$ such that
$\int_\tau^T u(s) ds \leq 1-Y$. Hence, the investor enters the contract
at time $\tau$ and can spend a remaining total volume of $1-Y$ up to maturity $T$. The
corresponding value of the optimization problem is
$$
\bar J(\tau, Y):=\esssup_{u \in U(\tau,Y)} E\left[\left.\int_\tau^T u(s)X(s)ds \right|\mathcal{F}_\tau\right].
$$

The main result in \cite{BD} states (roughly speaking) that a good
version $(J(t,y), \;t\in[0,T],\; y\in(-\infty,1])$ of the adapted random field $(\bar J(t,y), \;t\in[0,T],\; y\in(-\infty,1])$
is characterized as the unique solution to \cb{the }first order backward stochastic partial differential equation (BSPDE)
\begin{eqnarray*}
 J(t,y)&=&E\left[\left. L \int_t^T (X(s)+D^-_yJ(s,y))_+ds\right|\mathcal{F}_t\right], \nonumber \\
J(t,1)&=&0,
\end{eqnarray*}
which is smooth enough to ensure that the differential inclusion
\begin{equation*}
 u(s) \in \left\{\begin{array}{cl} \{0\}, & X(s)+D_y^-J(s,y+\int_t^s  u(r)dr)<0\\ \{L\}, & X(s)+D_y^-J(s,y+\int_t^s u(r)dr)>0 \\
 \ [0,L], &
X(s)+D_y^-J(s,y+\int_t^s  u(r)dr)=0. \end{array}\right.
\end{equation*}
has a solution $u\in U(t,y)$. Here, $D^-_y$ denotes the lefthand side derivative in the $y$-variable and $(\cdot)_+$ denotes the positive part.

The main purpose of the present paper is to study regularity of (the good version $J(t,y)$ of) the value process in the $y$-variable
and to replace the above smoothness condition in terms of the differential inclusion by a classical differentiability condition. To this end we
shall assume that $X$ is additionally \emph{left-continuous in expectation} (LCE), i.e. for every $[0,T]$-valued  stopping time $\sigma$ and every nondecreasing sequence of $[0,T]$-valued stopping times $(\sigma_n)_{n\in \mathbb{N}}$
with  limit $\sigma$ it holds that
\begin{equation}\label{LCE}
\lim_{n\rightarrow \infty} E[X(\sigma_n)]=E[X(\sigma)].
\end{equation}
Intuitively this means that the jumps of $X$ occur at total surprise and cannot be predicted.
Under this assumption we are going to prove the following theorem:
\begin{thm}\label{thm:main_new}
Suppose the standing assumptions and that $X$ \cb{is }left-continuous in expectation.
For every $t\in[0,T]$ denote
$$
\Delta_t:=(1-L(T-t),1),\quad \bar \Delta_t:=[1-L(T-t),1].
$$
Then, \\
 (i) There is a \cb{measurable }version $(J(t,y),\; t\in [0,T], y \in \bar\Delta_t)$ of $(\bar J(t,y),\; t\in [0,T], y \in \bar\Delta_t)$  which fulfills:
\begin{enumerate}
 \item[a)] There is a set $\bar \Omega \in \mathcal{F}$ with $P(\bar \Omega)=1$ such that $D^-_y J(t,\omega,y)$ exists  for every $t\in [0,T]$, $y\in \bar\Delta_t$, and $\omega \in \bar \Omega$
and is left-continuous in $y$. Moreover, $J$ is Lipschitz in $y$ in the following sense: There is an
integrable random variable $C$ satisfying
$$
|J(t,\omega,y_1)-J(t,\omega,y_2)|\leq C(\omega) |y_1-y_2|
$$
for every $t\in [0,T], \omega \in \bar \Omega$ and $y_1,y_2\in \bar\Delta_t$.
\item[b)] For every $t\in [0,T]$, there is a set $\Omega_t$ of full
$P$-measure such that, for every $\omega \in \Omega_t$, the mapping
$$
y\mapsto J(t,\omega,y)
$$
is continuously differentiable on  $\Delta_t$.
\item[c)] \cb{For every $t\in[0,T]$ and $y\in\Delta_t$, $\frac{\partial}{\partial y}J(s,\omega,y)$ exists for  $\lambda_{[t,T-\frac{1-y}{L}]}\otimes P$-almost every $(s,\omega)$,
\begin{equation}\label{BSPDE_new}
 J(t,y)=E\left[\left.\int_{T-\frac{1-y}{L}}^T LX(s)ds+ L \int_t^{T-\frac{1-y}{L}} (X(s)+\frac{\partial}{\partial y}J(s,y))_+ds\right|\mathcal{F}_t\right]
\end{equation}
holds $P$-almost surely }and
the boundary conditions
\begin{equation}\label{boundary}
J(t,1-L(T-t))=E\left[\left.  \int_t^T LX(s)ds \right|\mathcal{F}_t\right],\quad J(t,1)=0
\end{equation}
are satisfied.
\end{enumerate}
(ii) Conversely, if $(J(t,y),\; t\in [0,T], y \in \bar \Delta_t)$  is a \cb{measurable }random field satisfying a), b), and c), then it is a version of $\bar J$, i.e.
for every $t\in [0,T],\;y\in \bar \Delta_t$
$$
J(t,y)=\bar J(t,y)
$$
$P$-almost surely.
\end{thm}
\cb{In the above theorem and for the remainder of the paper $\lambda$ denotes the Lebesgue measure and $\lambda_{[a,b]}$ its restriction to the interval $[a,b]$.}

We note that the above theorem characterizes the value process on the set $\{(t,y);\; 0\leq t\leq T, 1-L(T-t)\leq y \leq 1\}$. For $y\leq 1-L(T-t)$ the optimization becomes trivial,
because the remaining volume $1-y$ \cb{is }at least as large as the maximal volume $L(T-t)$ which one can spend when exercising at the maximal rate of $L$. Hence, $L{\bf 1}_{[t,T]}$ is an optimal strategy
and
\begin{equation}\label{eq:boundaryL}
\bar J(t,y)=E\left[\left.\int_t^T LX(s)ds \right|\mathcal{F}_t\right].
\end{equation}
This also explains the boundary condition (\ref{boundary}) at $y=1-L(T-t)$.

We emphasize that condition b) in Theorem \ref{thm:main_new} is a classical $\mathcal{C}^1$-condition on the solution of the BSPDE. Hence, we can interpret this theorem as an existence and uniqueness
result of a classical solution for the BSPDE (\ref{BSPDE_new})--(\ref{boundary}). Taking  into account that  this BSPDE is a non-Markovian version of a Hamilton-Jacobi-Bellman equation, we think that
existence of a classical solution is a striking feature. Indeed, recent studies of the corresponding HJB equation for stochastic control problems with integral constraints
in the Markovian diffusion case
such as \citet{BCV} only discuss
the HJB equation in the framework of viscosity solutions.

The paper is organized as follows: In Section 2 we give a recap of some of the results in \citet{BD}, to which we refer as [BD] from now on. The proof
 of Theorem~\ref{thm:main_new} is divided into two parts.
In Section 3 we prove the uniqueness part, i.e. we show that every
adapted random field which satisfies a), b), and c) coincides necessarily with the value process $\bar J(t,y)$. It will turn out that the LCE assumption is not required for this part of
Theorem \ref{thm:main_new}.
It is however crucial for the smoothness part which is proved in Section 3. Here we show that a good version of the value process is indeed continuously differentiable in the sense of
b). The derivative in the space variable is additionally represented via some nonstandard optimal stopping problems which can be linked to the interpretation of the derivative as the marginal
value of the underlying control problem. Finally, in Section 4 we derive a dual minimization over martingales and relate the minimizing martingale to the derivative of the value process.

\section{Recap of the main results in [BD]}

In this section we state some results from [BD] for handy reference. We recall that the standing assumptions are in force without further mention.

The first result, Proposition 3.5 in [BD], provides a good version of the value process $\bar J$.

\begin{prop}\label{prop:goodversion}
There is an adapted random field $(J(t,y),\; t\in[0,T],\, y\in(-\infty,1])$ such that
$$
J(\tau,Y)=\bar J(\tau,Y)\quad P-a.s.
$$
for every $[0,T]$-valued stopping time $\tau$ and every $\mathcal{F}_\tau$-measurable, $(-\infty,1]$-valued random variable $Y$. Moreover, $J$  satisfies the following:
\\
 There is a set $\bar \Omega \in \mathcal{F}$ with $P(\bar \Omega)=1$ such that  the following properties hold on $\bar \Omega$:
\begin{enumerate}
 \item For every $y\in(-\infty,1]$, the mapping $t\mapsto  J(t,y)$ is RCLL.
  \item For every $t\in[0,T]$ and $y_1,y_2\in (-\infty,1]$
 $$
|J(t,y_1)-J(t,y_2)|\leq \left(\sup_{r\in [0,T]} Z(r) \right) |y_1-y_2|
$$
where $Z(t)$ is a RCLL modification of $E[\sup_{r\in [0,T]} X(r)|\mathcal{F}_t]$ which satisfies $$\sup_{r\in [0,T]} Z(r) < \infty$$
on $\bar \Omega$.
\item For every $t\in[0,T]$, the mapping $y\mapsto J(t,y)$ is concave.
\end{enumerate}
\end{prop}

The main theorem of [BD] characterizes the value process. It does not require the LCE assumption.

\begin{thm}\label{thm:main_old}
 (i) The version $(J(t,y),\; t\in [0,T], y \in (-\infty,1])$ of $(\bar J(t,y),\; t\in [0,T], y \in (-\infty,1])$  constructed in Proposition \ref{prop:goodversion} satisfies:
\begin{enumerate}
 \item[a')] There is a set $\bar \Omega \in \mathcal{F}$ with $P(\bar \Omega)=1$ such that $D^-_y J(t,\omega,y)$ exists  for every $t\in [0,T],\; y\in (-\infty,1]$ and $\omega \in \bar \Omega$
and is left-continuous in $y$. Moreover, $J$ is Lipschitz in $y$ in the following sense: There is an
integrable random variable $C$ satisfying
$$
|J(t,\omega,y_1)-J(t,\omega,y_2)|\leq C(\omega) |y_1-y_2|
$$
for every $t\in [0,T], \omega \in \bar \Omega$ and $y_1,y_2\in (-\infty,1]$.
\item[b')] For every $(t,y)\in[0,T]\times (-\infty,1]$, there is a control $u^{t,y}\in U(t,y)$ such that the differential inclusion
\begin{equation}\label{eq:optimality}
 u^{t,y}(s) \in \left\{\begin{array}{cl} \{0\}, & X(s)+D_y^-J(s,y+\int_t^s  u^{t,y}(r)dr)<0\\ \{L\}, & X(s)+D_y^-J(s,y+\int_t^s u^{t,y}(r)dr)>0 \\
 \ [0,L], &
X(s)+D_y^-J(s,y+\int_t^s  u^{t,y}(r)dr)=0. \end{array}\right.
\end{equation}
is satisfied  $\lambda_{[t,T]}\otimes P$-almost surely.
\item[c')] For every $(t,y)\in[0,T]\times (-\infty,1]$,
\begin{eqnarray*}
 J(t,y)&=&E\left[\left. L \int_t^T (X(s)+D^-_yJ(s,y))_+ds\right|\mathcal{F}_t\right], \nonumber \\
J(t,1)&=&0,
\end{eqnarray*}
$P$-almost surely.
\end{enumerate}
(ii) Conversely, if $(J(t,y),\; t\in [0,T], y \in (-\infty,1])$  is a \cb{measurable }random field satisfying a'), b'), and c'), then it is a version of $\bar J$, i.e.
for every $(t,y)\in[0,T]\times (-\infty,1]$
$$
J(t,y)=\bar J(t,y)
$$
$P$-almost surely. In this case, $u^{t,y}\in U(t,y)$ is optimal for $\bar J(t,y)$, if and only if  (\ref{eq:optimality}) is satisfied.
\end{thm}

Theorem \ref{thm:main_old} includes an existence result for optimal controls. One can even choose an optimal control with some additional
properties which turns out to be useful later.

\begin{prop}\label{prop:optimalcontrol}
 For every $[0,T]$-valued stopping time $\tau_0$ and every $\mathcal{F}_{\tau_0}$-measurable, $(-\infty,1]$-valued random variable $Y_0$, there
is an optimal strategy $u^{\tau_0,Y_0}\in U(\tau_0,Y_0)$ for $\bar J(\tau_0,Y_0)$ such that
\begin{equation}\label{U'}
 u^{\tau_0,Y_0}(r)=L \textnormal{ on } \{L(T-r)\leq 1-(Y_0+\int_{\tau_0}^r u^{\tau_0,Y_0}(s)ds)\}.
\end{equation}
and
\begin{equation}\label{U'2}
 \int_{\tau_0}^T  u^{\tau_0,Y_0}(s)ds=1-Y \textnormal{ on } \{L(T-\tau_0)\geq1-Y_0\}
\end{equation}
\end{prop}
The above proposition actually is a direct consequence of Proposition 3.2 in [BD].

As a corollary to Theorem \ref{thm:main_old} we observe that it is optimal to exercise a submartingale as late as possible and a supermartingale as early as possible,
which is as expected.
\begin{cor}\label{cor:submartingale}
 Suppose $X$ satisfies the standing assumptions,  $\tau_0$ is a $[0,T]$-valued stopping time and $Y_0$ is a $\mathcal{F}_{\tau_0}$-measurable, $(-\infty,1]$-valued random variable.
\\[0.1cm] (i) If $X$ is an RCLL submartingale, then
$ u^{\tau_0,Y_0}=L {\bf 1}_{[(T-(1-Y_0)/L)\vee \tau_0,T]}$ is optimal for $\bar J(\tau_0,Y_0)$.
\\[0.1cm] (ii) If  $X$ is an RCLL supermartingale, then
$ u^{\tau_0,Y_0}=L {\bf 1}_{[\tau_0,(\tau_0+(1-Y_0)/L)\wedge T]}$ is optimal for $\bar J(\tau_0,Y_0)$.
\end{cor}
This corollary is indicated at the end of Example 2.2 in [BD]. For sake of completeness we here provide a proof.
\begin{proof}
We only prove the submartingale case, as the supermartingale case is similar.
For any  $[0,T]$-valued stopping time $\tau_0$ and any  $\mathcal{F}_{\tau_0}$-measurable, $(-\infty,1]$-valued random variable $Y_0$ we define
$$
\bar V(\tau_0,Y_0)=E\left[\left.  \int_{\tau_0}^T u^{\tau_0,Y_0}(s)X(s)ds \right|\mathcal{F}_t\right]=E\left[\left.  \int_{\tau_0\vee(T-(1-Y_0)/L)}^T LX(s)ds \right|\mathcal{F}_t\right].
$$
Denote the good version of the value process constructed in Proposition \ref{prop:goodversion}  by $J$. Then,
$$
J(\tau_0,Y_0)\geq  \bar V(\tau_0,Y_0)
$$
$P$-almost surely.
Hence, $u^{\tau_0,Y_0}$ is optimal if and only if
\begin{equation}\label{hilf0003}
 0=E[J(\tau_0,Y_0)-\bar V(\tau_{\nd{0}},Y_0) ]=E\left[J(\tau_0,Y_0)-\int_{\tau_0\vee(T-(1-Y_0)/L)}^T LX(s)ds\right].
\end{equation}
Notice that, if (\ref{hilf0003}) holds for all  deterministic pairs $(\tau_0,Y_0)$, then it is also true for general pairs. Indeed,
$J(\tau_0,Y_0)= \bar V(\tau_0,Y_0)$
then  holds $P$-almost surely for pairs $(\tau_0, Y_0)$
which take at most countably many values. By the continuity properties of $J$ one can then pass to the limit to obtain (\ref{hilf0003}) for general pairs $(\tau_0, Y_0)$.

It is hence sufficient to show optimality of $u^{t,y}$ for deterministic $t\in [0,T]$ and $y\in (-\infty,1]$.
 By similar arguments than at the beginning of the proof of Proposition 3.5 in [BD], there is a version $V(t,y)$ of
$\bar V(t,y)$ which satisfies  a') in Theorem \ref{thm:main_old}. By the definition of $\bar V$ it is straightforward that
$$
D^-_y V(t,y)=\left\{\begin{array}{cl} - E[X((T-(1-y)/L)-)|\mathcal{F}_t], & t< T-(1-y)/L \\ 0, & t\geq T-(1-y)/L, \end{array}\right.
$$
where $X(s-)$ denotes the left limit of $X$ at $s$. Hence, by the submartingale property,
$
X(s)+D^-_y V(s,y) \leq 0
$
for $s< T-(1-y)/L$ and
$
X(s)+D^-_y V(s,y)=X(s) \geq 0
$
for  $s\geq T-(1-y)/L$. Thus, $V$ solves the BSPDE in c') of Theorem \ref{thm:main_old}. It is also straightforward to see that $u^{t,y}=L {\bf 1}_{[t\vee (T-(1-y)/L),T]}$
solves the differential inclusion in b') of Theorem \ref{thm:main_old} with $V$ in place of $Y$, because
$$
 X(s)+D_y^-V(s,y+\int_t^s  u^{t,y}(r)dr)=\left\{\begin{array}{cl} X(s)- E[X((T-(1-y)/L)-)|\mathcal{F}_s]\leq 0, & s< T-(1-y)/L \\ X(s)\geq 0, & s\geq T-(1-y)/L. \end{array}\right.
$$
Consequently, by Theorem \ref{thm:main_old}, $u^{t,y}$ is optimal.
\end{proof}

\section{Uniqueness of classical solutions}

This section is devoted to the proof of the uniqueness part  of Theorem \ref{thm:main_new}. It relies on Theorem \ref{thm:main_old}, (ii).
 This is what we are actually going to show:
\begin{thm}\label{thm:uniqueness}
 Under the standing assumption, suppose that $(J(t,y),\; t\in [0,T], y \in \bar \Delta_t)$  is a \cb{measurable }random field satisfying a), b), and c) of Theorem \ref{thm:main_new}.
Define $$J(t,y):=J(t,1-L(T-t)) ,\quad t\in [0,T],\; y<1-L(T-t).$$
Then
$(J(t,y),\; t\in [0,T], y \in (-\infty,1])$
fulfills conditions  a'), b'), and c') of Theorem \ref{thm:main_old}. In particular,
it is a version of $\bar J$, i.e.
for every $(t,y)\in[0,T]\times (-\infty,1]$
$$
J(t,y)=\bar J(t,y)
$$
$P$-almost surely.
\end{thm}

For the remainder of this Section we assume that $(J(t,y),\; t\in [0,T], y \in \bar \Delta_t)$  is a \cb{measurable }random field satisfying a), b), and c) of Theorem \ref{thm:main_new},
and that it is extended to $[0,T]\times (-\infty,1]$ as described in Theorem \ref{thm:uniqueness}. We first verify conditions a') and c') of Theorem \ref{thm:main_old}.

\begin{lem}\label{lem:ac}
 $J$ satisfies a') and c') in Theorem \ref{thm:main_old}.
\end{lem}
\begin{proof}
Property a') is a direct consequence of a) in Theorem \ref{thm:main_new} and the constant extrapolation of $J$.
 Property  c') holds for $t\in [0,T]$ and
$1-L(T-t)<y<1$ by property c) in Theorem \ref{thm:main_new}, \cb{noting that $D^-_yJ(s,y)=0$ for $s\geq T-(1-y)/L$ by the constant extrapolation}. It then extends to $y=1$ by the continuity properties in a'). Finally, for $y\leq 1-L(T-t)$,
$$
J(t,y)=J(t,1-L(T-t))=E\left[\left.  \int_t^T LX(s)ds \right|\mathcal{F}_t\right]=E\left[\left. L \int_t^T (X(s)+D^-_yJ(s,y))_+ds\right|\mathcal{F}_t\right],
$$
where we used the boundary condition (\ref{boundary}) and \cb{again }the fact that $D^-_yJ(t,y)=0$ for  $y\leq 1-L(T-t)$.
\end{proof}

In order to prove  that the differential inclusion (\ref{eq:optimality}) has a solution, and hence, $J$ satisfies condition  b') in Theorem \ref{thm:main_old},
we denote
\begin{eqnarray*}
 \Gamma_+(t,\omega)&=&\{y\in (1-L(T-t),1);\; X(t,\omega)+D_y^-J(t,\omega,y)>0 \}\\
\Gamma_-(t,\omega)&=&\{y\in (1-L(T-t),1);\; X(t,\omega)+D_y^-J(t,\omega,y)<0 \} \\
\Gamma_0(t,\omega)&=&\mathbb{R} \setminus (\Gamma_+(t,\omega) \cup \Gamma_-(t,\omega)).
\end{eqnarray*}
As a preparation we first prove the following lemma.

\begin{lem}
 For every $(t_0,y_0)\in [0,T]\times \mathbb{R}$ there is an adapted process $\hat u$ such that
\begin{equation}\label{eq:diffincl1}
 \hat u(s) \in \left\{\begin{array}{cl} \{0\}, & y_0+\int_{t_0}^s  \hat u(r)dr\in \Gamma_-(s)\\ \{L\},
&  y_0+\int_{t_0}^s  \hat u(r)dr\in \Gamma_+(s) \\
 \ [0,L], &
 y_0+\int_{t_0}^s  \hat u(r)dr\in \Gamma_0(s) \end{array}\right.
\end{equation}
$\lambda_{[t_0,T]}\otimes P$-almost surely.

\end{lem}
\begin{proof}
We apply a somewhat standard technique
approximating the differential inclusion by a sequence of differential equations with Lipschitz  coefficients, see e.g. the textbook by \citet{AC}.
The crucial observation is that we can construct the sequence of approximating differential equations in a monotonic way, which leads to almost sure convergence
(while without monotonicity one obtains
weak $L^2$-convergence only).

 To this
end define for $m\in \mathbb{N}$
$$
 \Gamma_{+,m}(t,\omega)=\{y\in (1-L(T-t),1);\; X(t,\omega)+D_y^-J(t,\omega,y)\geq 1/m \},
$$
and
for $n\in \mathbb{N}$, $s\in [0,T]$, and $y\in \mathbb{R}$
$$
F_n(s,\omega,y)=2^n\int_{y-2^{-n}}^y  L{\bf 1}_{\{v;\;\exists {m(\omega)\in\mathbb N}\; \forall \eta \in [0,2^{-n}] \;\; v+\eta \in \Gamma_{+,m(\omega)}(s,\omega)\}} dv.
$$
Notice that due to the leftcontinuity assumption on $D^-_yJ$
\begin{eqnarray*}
 && \exists {m(\omega)\in\mathbb N}\;\forall \eta \in [0,2^{-n}] \;\; v+\eta \in \Gamma_{+,m(\omega)}(s,\omega)\\
&\Leftrightarrow& v\in (1-L(T-s),1-2^{-n}) \textnormal{ and } \inf_{\eta \in [0,2^{-n}]\cap \mathbb{Q}} X(s,\omega)+D^-_yJ(s,\omega,v+\eta)>0.
\end{eqnarray*}
Hence, the integrand in the definition of $F_n$ is measurable (as function in $(s,\omega,v)$) and, in particular, $F_n(\cdot,y)$ is \cb{an $(\mathcal{F}_s)_{s\in[t_0,T]}$-adapted process  }for every $y\in \mathbb{R}$. Moreover, by construction,
$F_n$ is Lipschitz in $y$ with constant $L2^n$ (independent of $s,\omega$).
Hence there is a unique (up to indistinguishability) \cb{continuous and }adapted process $y_n$ which satisfies
$$
y_n(s)=y_0+\int_{t_0}^s F_n(r, y_n(r)) dr,\quad s\in [t_0,T].
$$
We define the sequence of $[0,L]$-valued adapted process $(u_n)_{n\in \mathbb{N}}$ via
$$
u_n(s)=F_n(s, y_n(s)),\quad  s\in [t_0,T].
$$
This sequence belongs to the set of adapted and $[0,L]$-valued processes, which as a subset of  $L^2([\cb{t_0},T]\times \Omega)$ is bounded, closed and convex, and, thus, weakly compact.
Consequently, there is an adapted $[0,L]$-valued process $\hat u$ and a subsequence $(n_k)$ such that
$$
u_{n_k} \rightarrow \hat u,\quad k\rightarrow \infty
$$
weakly in $L^2([t_0,T]\times \Omega)$. We claim that $\hat u$ satisfies  (\ref{eq:diffincl1}). In order to see this, fix $s\in [t_0,T]$ and choose
an arbitrary
$\xi\in L^2(\Omega)$. Then, $\xi {\bf 1}_{[t_0,s]} \in L^2([t_0,T]\times \Omega)$. Hence,
\begin{eqnarray*}
&& E[\xi (y_{n_k}(s)-y_0)]= E[\int_{t_0}^T  \xi {\bf 1}_{[t_0,s]}(r) u_{n_k} (r) dr]\\ &\rightarrow & E[\int_{t_0}^T  \xi {\bf 1}_{[t_0,s]}(r) \hat u (r) dr]
= E[\xi \int_{t_0}^s \hat u(r) dr],
\end{eqnarray*}
i.e. $y_{n_k}(s)$ converges to $y_0+\int_{t_0}^s \hat u(r) dr$ weakly in $L^2(\Omega)$.
We shall now show that this convergence holds almost surely, indeed. To this end we first observe that $F_n(s,y)\leq F_{n+1}(s,y)$ for every pair $(s,y)$, because
\begin{eqnarray*}
 F_n(s,y)&=&2^n\int_{y-2^{-(n+1)}}^y   L{\bf 1}_{\{v;\;\exists {m\in\mathbb N}\, \forall \eta \in [0,2^{-n}] \; v-2^{-(n+1)}+\eta \in \Gamma_{+,m}(s)\}} dv\\ &&+2^n\int_{y-2^{-(n+1)}}^y
 L{\bf 1}_{\{v;\;\exists {m\in\mathbb N}\, \forall \eta \in [0,2^{-n}] \; v+\eta \in \Gamma_{+,m}(s)\}} dv\\
&\leq& 2\cdot 2^n \int_{y-2^{-(n+1)}}^y  L{\bf 1}_{\{v;\;\exists {m\in\mathbb N}\, \forall \eta \in [0,2^{-(n+1)}] \; v+\eta \in \Gamma_{+,m}(s)\}} dv=F_{n+1}(s,y).
\end{eqnarray*}
As $y_{n+1}-y_n$ is (up to indistinguishability) the unique continuous solution of the linear differential equation
$$
y_{n+1}(s)-y_n(s)=\int_{t_0}^s (b_n(r)(y_{n+1}(r)-y_n(r))+c_n(r))dr
$$
for
\begin{eqnarray*}
 b_n(r)&=&{\bf 1}_{\{y_{n+1}(r)\neq y_n(r)\}} \frac{F_{n+1}(r,y_{n+1}(r))-F_{n+1}(r,y_{n}(r))}{y_{n+1}(r)- y_n(r)} \\
c_n(r)&=& F_{n+1}(r,y_{n}(r))-F_{n}(r,y_{n}(r))\geq 0
\end{eqnarray*}
we thus observe that, on a set of full $P$-measure independent of $n,s$,
$$
y_{n+1}(s)-y_n(s)= \int_{t_0}^s c_n(r) \exp\{\int_{r}^s b_n(u)du\}dr \geq 0,
$$
i.e. the sequence $y_n(s)$ is nondecreasing. As it is bounded by $y_0+L(T-t_0)$, the $P$-a.s. limit
$$
y(s)=\lim_{n\rightarrow \infty} y_n(s)
$$
exists for every $s\in[t_0,T]$. By dominated convergence $(y_n(s))$ converges to $y(s)$ strongly in $L^2(\Omega)$, and in view of the weak convergence obtained above
we conclude that for every $s\in[t_0,T]$
$$
y(s)=y_0+\int_{t_0}^s \hat u(r) dr,\quad \cb{P\textnormal{-a.s.}}
$$
We next introduce the $\mathcal{B}_{[t_0,T]}\otimes \mathcal{F}$-measurable set
$$
\cb{A}=\{(s,\omega);\ y(s,\omega)\in \Gamma_+(s,\omega) \textnormal{ and } \inf_{\eta \in [-2^{-n},2^{-n}]\cap \mathbb{Q}} X(s,\omega)+D^-_yJ(s,\omega,y_n(s,\omega)+\eta)\leq 0 \textnormal{ i.o.}\},
$$
where `i.o.' means that the event occurs for infinitely many $n$'s. Denote its $s$-section for $s\in [t_0,T]$ by
$$
A_s=\{\omega;\; (s,\omega)\in \cb{A}\}.
$$
As $y_n(s)$ converges to $y(s)$, we observe that $A_s\subset \Omega_s^c$, where $\Omega_s$ is the set of full $P$-measure in property   b) of Theorem \ref{thm:main_new},
on which $y\mapsto D^-_yJ(s,y)$ is continuous. Hence, $P(A_s)=0$.
Consequently, by weak convergence of $(u_{n_k})$ to $\hat u$,
\begin{eqnarray}\label{hilf0009}
 E[\int_{t_0}^T (L-\hat u(s)){\bf 1}_{\{y(s)\in \Gamma_+(s)\}}ds]
&=& \lim_{k\rightarrow \infty} E[\int_{t_0}^T (L-u_{n_k}(s)) {\bf 1}_{\{y(s)\in \Gamma_+(s)\}}{\bf 1}_{A_s^c}ds].
\end{eqnarray}
We now notice that \cb{for every $s\in[t_0,T]$ }(using convergence of $y_n(s)$ to $y(s)$ for the second identity
to conclude that $y_n(s)\in (1-L(T-s)+2^{-n},1-2^{-n})$ for sufficiently large $n$)
\begin{eqnarray*}
&& \{y(s)\in \Gamma_+(s)\}\cap A_s^c\\ &=& \{y(s)\in \Gamma_+(s)\}\cap\left( \bigcup_{N\in \mathbb{N}}\bigcap_{n\geq N} \left\{\inf_{\eta \in [-2^{-n},2^{-n}]\cap \mathbb{Q}} X(s)+D^-_yJ(s,y_n(s)+\eta)> 0\right\}\right)
\\ &=& \{y(s)\in \Gamma_+(s)\}\cap\left( \bigcup_{N\in \mathbb{N}}\bigcap_{n\geq N} \left\{\exists {m\in\mathbb N}\, \forall \eta \in [-2^{-n},2^{-n}] \; y_n(s)+\eta \in \Gamma_{+,m}(s)\right\}\right)
\\ &\subset& \{y(s)\in \Gamma_+(s)\}\cap\left( \bigcup_{N\in \mathbb{N}}\bigcap_{n\geq N} \{F_n(s,y_n(s))=L\} \right)\\ &=& \{y(s)\in \Gamma_+(s)\}\cap\left( \bigcup_{N\in \mathbb{N}}\bigcap_{n\geq N} \{u_n(s)=L\} \right)
\end{eqnarray*}
Thus, by dominated convergence, the right-hand side of (\ref{hilf0009}) converges to zero.
An analogous argument shows
$$
E[\int_{t_0}^T \hat u(s){\bf 1}_{\{y(s)\in \Gamma_-(s)\}}ds]=0.
$$
As $\hat u$ is $[0,L]$-valued these two identities imply that \cb{$\hat u$ }solves the differential inclusion (\ref{eq:diffincl1}).
\end{proof}

\begin{prop}
$J$ fulfills b') of Theorem \ref{thm:main_old}, i.e. for every $(t,y)\in [0,T]\times (-\infty,1]$, there is a $u^{t,y}\in U(t,y)$ which satisfies (\ref{eq:optimality}).
\end{prop}

\begin{proof}
 We first denote by  $\hat u$ the solution of (\ref{eq:diffincl1}) with $(t_0,y_0)=(t,y)$ constructed in  the previous lemma. Define
\begin{eqnarray*}
\bar\sigma_{\upper}&:=&\inf\{r\geq t;\; y+\int_{t}^r \hat u(s)ds\geq 1\} \wedge T , \\
\bar\sigma_L&:=&\inf\{r\geq t;\; y+\int_{t}^r \hat u(s)ds\leq 1-L(T-t)\} \wedge T,
\end{eqnarray*}
and
$$
u^{t,y}(r):=\left\{\begin{array}{cl} \hat u(r), & r\in [t,\bar\sigma_{\upper}\wedge \bar\sigma_L) \\
 L{\bf 1}_{\{\bar\sigma_L<\bar\sigma_{\upper}\}} & r \in  [\bar\sigma_{\upper}\wedge\bar\sigma_L,T] \end{array}\right.
$$
Let $y^{t,y}(r)=y+\int_t^r u^{t,y}(s)ds$ \cb{and $y(r)=y+\int_t^r \hat u(s)ds$}. We wish to show that $u^{t,y}$ solves (\ref{eq:optimality}), for which it suffices to verify that
\begin{eqnarray}\label{eq:hilf0001}
 E\left[\int_t^T |u^{t,y}(s)|\,{\bf 1}_{\{X(s)+D_y^-J(s,y^{t,y}(s))<0\}}+|u^{t,y}(s)-L|\,{\bf 1}_{\{X(s)+D_y^-J(s,y^{t,y}(s))>0\}}\,ds\right]=0.
\end{eqnarray}
We decompose
\begin{eqnarray*}
&&  E\left[\int_t^T |u^{t,y}(s)|\,{\bf 1}_{\{X(s)+D_y^-J(s,y^{t,y}(s))<0\}}+|u^{t,y}(s)-L|\,{\bf 1}_{\{X(s)+D_y^-J(s,y^{t,y}(s))>0\}}\,ds\right]\\
&=& E\left[\int_t^T |\hat u(s)|\,{\bf 1}_{\{X(s)+D_y^-J(s, y(s))<0\}} \,{\bf 1}_{\{s< \bar\sigma_{\upper}\wedge \bar\sigma_L\}} ds\right] \\ && +
E\left[\int_t^T L \,{\bf 1}_{\{X(s)+D_y^-J(s,y^{t,y}(s))<0\}} \,{\bf 1}_{\{s\geq \bar\sigma_{\upper}\wedge \bar\sigma_L\}} {\bf 1}_{\{\bar\sigma_L<\bar\sigma_{\upper}\}} ds\right]
\\ && + E\left[\int_t^T |\hat u(s)-L|\,{\bf 1}_{\{X(s)+D_y^-J(s, y(s))>0\}} \,{\bf 1}_{\{s< \bar\sigma_{\upper}\wedge\bar\sigma_L\}} ds\right]
\\ && +
E\left[\int_t^T L \,{\bf 1}_{\{X(s)+D_y^-J(s,y^{t,y}(s))>0\}} \,{\bf 1}_{\{s\geq \bar\sigma_{\upper}\wedge\bar\sigma_L\}} {\bf 1}_{\{\bar\sigma_L\geq \bar\sigma_{\upper}\}} ds\right]=(I)+(II)+(III)+(IV).
\end{eqnarray*}
The previous lemma implies that the first and the third term \cb{vanish. For }the second term we notice that $y^{t,y}(s)\leq 1-L(T-s)$ on $\{s\geq \bar\sigma_{\upper}\wedge\bar\sigma_L\}\cap
\{\bar\sigma_L<\bar\sigma_{\upper}\}$ by the definition of $\bar\sigma_L$ and $u^{t,y}$. But then $D_y^-J(s,y^{t,y}(s))=0$ by the constant extrapolation of $J$.  Hence,
$$
(II)\leq  E\left[\int_t^T L \,{\bf 1}_{\{X(s)<0\}} ds\right]=0,
$$
because of  the nonnegativivity of $X$.
Similarly one can treat the fourth term. We first observe by c') of Theorem \ref{thm:main_old} that
$$
0=J(t,1)=E\left[\left. L \int_t^T (X(s)+D^-_yJ(s,1))_+ds\right|\mathcal{F}_t\right],
$$
which yields
$$
X(s)+D_y^-J(s,1)\leq 0,\quad \lambda_{[t,T]}\otimes P\textnormal{-a.s.}
$$
However, we have  $y^{t,y}(s)=1$ on $\{s\geq \bar\sigma_{\upper}\wedge\bar\sigma_L\}\cap
\{\bar\sigma_L\geq \bar\sigma_{\upper}\}$ by the definition of $\bar\sigma_{\upper}$ and $u^{t,y}$. Hence,
$$
(IV)\leq   E\left[\int_t^T L \,{\bf 1}_{\{X(s)+D_y^-J(s,1)>0\}} ds\right]=0.
$$
We finally note that  $$\int_t^T u^{t,y}(r) dr= y^{t,y}(T)-y \leq 1-y.$$ Thus,
$u^{t,y}$ belongs to $U(t,y)$.
\end{proof}

In view of Theorem \ref{thm:main_old}, (ii), and Lemma \ref{lem:ac}, the above proposition concludes the proof of Theorem \ref{thm:uniqueness},
and hence of the uniqueness part (ii) of Theorem \ref{thm:main_new}.

\section{Regularity of the value process}

In this section we study regularity of the `good' version $J$ of the value process in the $y$-variable.
For the remainder of this section we always assume that $J$ is the random field constructed in Proposition \ref{prop:goodversion}. Notice that, by concavity,
the one-sided derivatives $D^{\pm}_yJ(t,y)$ exist.
In view of Theorem \ref{thm:main_old} and (\ref{eq:boundaryL}) we observe that $J$ satisfies the requirements of Theorem \ref{thm:main_new}, (i),  once we establish
the following result.
 \begin{thm}\label{thm:smooth}
 Suppose that $X$ satisfies the standing assumptions and it is LCE. Then, for every $t\in [0,T]$, there is a set $\Omega_t$ of full
$P$-measure such that, for every $\omega \in \Omega_t$, the mapping
$$
y\mapsto J(t,\omega,y)
$$
is continuously differentiable on  $(1-L(T-t),1)$. \cb{Moreover, for every $t\in[0,T]$ and $y\in (1-L(T-t),1)$, $\frac{\partial}{\partial y}J(s,\omega,y)$ exists
for $\lambda_{[t,T-\frac{1-y}{L}]}\otimes P$-almost every $(s,\omega)$.}
\end{thm}

\begin{rem}
(i) The LCE assumption is crucial for Theorem \ref{thm:smooth} to hold. Example 4.5 in [BD] provides a counterexample to the assertion of this theorem
for a process $X$ which fails to be LCE. \\[0.1cm]
 (ii) Theorem \ref{thm:smooth2} below implies that, under the assumptions of Theorem \ref{thm:smooth}, the following stronger regularity assertion holds, if and only if
$X(T)=0$ $P$-almost surely:  For every $t\in [0,T]$, there is a set $\Omega_t$ of full
$P$-measure such that, for every $\omega \in \Omega_t$, the mapping
$$
y\mapsto J(t,\omega,y)
$$
is continuously differentiable on  $(-\infty,1)$.
\\ (iii) It was also noticed in the context of continuous time multiple stopping problems that the regularity of the value process
is typically improved, when $X(T)=0$, see \citet{Be2}.
\end{rem}
\cb{Recall that the derivatives $-D^{\pm}_yJ(t,y)$
correspond to the marginal value of the control problem. We first choose an optimal control $ u^{t,y}$ for $\bar J(t,y)$. Then heuristically,
in order to calculate $-D^-_yJ(t,y)\approx \frac{J(t,y-\Delta y)-J(t,y)}{\Delta y}$ one would like to spend optimally an inifinitesimal additional volume of $\Delta y$ at a time $\sigma$ where exercise is still
possible, i.e. where $u^{t,y}(\sigma)<L$. This intuitively leads to the following optimal stopping problem: maximize $E[X(\sigma)|\mathcal{F}_t]$ over stopping times $\sigma(\omega)$ which take values
in the set $\{s>t;\;  u^{t,y}(s,\omega)<L\}$. More precisely, the additional volume $\Delta y$ cannot be spent at a single time point $\sigma$, but rather
in `small' neighborhoods around $\sigma$.
Therefore we have to restrict the optimal stopping problem to time points $\sigma$, such that, given the strategy $ u^{t,y}$, it  is still possible
`to exercise in small neighborhoods of $\sigma$'.  This is how we make this heuristic idea precise.}

Fix a stopping time $\tau_0$ with values in $[0,T]$, an $\mathcal{F}_{\tau_0}$ measurable, $(-\infty,1]$-valued random variable  $Y_0$,
and an optimal control $u^{\tau_0,Y_0}\in U(\tau_0,Y_0)$ which satisfies the properties in Proposition \ref{prop:optimalcontrol}.
We denote
$$
M(\tau_0,Y_0):=\{L(T-\tau_0)>1-Y_0>0\} \in \mathcal{F}_{\tau_0},
$$
and
define
$$
A(\tau_0,Y_0):=\{t\in (\tau_0,T];\ \nd{\forall\,{\epsilon>0}}\;\, \lambda(\{s\in [\tau_0\vee (t-\epsilon), (t+\epsilon)\wedge T];\,
 u^{\tau_0,Y_0}(s)<L\})>0\}.
$$
This set $A(\tau_0,Y_0)$ is our way to make precise the set of time points \cb{$t$, such that given $u^{\tau_0,Y_0}$, one can still exercise in small neighborhoods of $t$}. We also introduce
$$
B(\tau_0,Y_0):=\{t\in (\tau_0,T];\ \nd{\forall\,{\epsilon>0}}\;\, \lambda(\{s\in [\tau_0\vee (t-\epsilon), (t+\epsilon)\wedge T];\,
u^{\tau_0,Y_0}(s)>0\})>0\},
$$
which corresponds to those points, where one can take away some marginal volume from the optimal control  $u^{\tau_0,Y_0}$.

We denote by $\mathcal{S}_{\tau_0+}^p$ the set of predictable stopping times $\sigma$ with values in $(\tau_0,T]\cup\{T\}$, by
$\mathcal{S}_{A(\tau_0,Y_0)}^p$ the set of stopping times $\sigma \in  \mathcal{S}_{\tau_0+}^p$ such that $\sigma$ takes values
 in $A(\tau_0,Y_0)$
on the set $M(\tau_0,Y_0)$, and by $\mathcal{S}_{B(\tau_0,Y_0)}^p$ the set of stopping times $\sigma \in  \mathcal{S}_{\tau_0+}^p$ such that $\sigma$ takes values
 in $B(\tau_0,Y_0)$
on the set $M(\tau_0,Y_0)$.

  The proof of Theorem \ref{thm:smooth} is prepared by several lemmas. We first  show that the set
$\mathcal{S}_{A(\tau_0,Y_0)}^p\cap\mathcal{S}_{B(\tau_0,Y_0)}^p$ is nonempty.
\begin{lem}\label{lem:barsigma}
 Suppose $X$ satisfies the standing assumptions and $(\tau_0,Y_0)$ are as above. Define
\begin{equation}\label{eq:barsigma}
\bar \sigma=\left\{\begin{array}{cl} \inf\{t\geq \tau_0;\; \int_{\tau_0}^t  u^{\tau_0,Y_0}(s)ds
\notin (1-Y_0-L(T-t),1-Y_0)\},& \omega \in M(\tau_0,Y_0) \\
                    T,& \textnormal{otherwise}.
 \end{array} \right.
\end{equation}
Then, $\bar \sigma \in \mathcal{S}_{A(\tau_0,Y_0)}^p\cap\mathcal{S}_{B(\tau_0,Y_0)}^p$.
\end{lem}
\begin{proof}
 On the set $M(\tau_0,Y_0)$ we have
$$
\int_{\tau_0}^T  u^{\tau_0,Y_0}(s)ds=1-\cb{Y_0}.
$$
by Proposition \ref{prop:optimalcontrol}. Moreover, as $0\in (1-Y_0-L(T-t),1-Y_0)$ on $M(\tau_0,Y_0)$, we get $\tau_0<\bar\sigma\leq T$ on $M(\tau_0,Y_0)$.
Hence, the stopping time $\bar\sigma$ takes values in $(\tau_0,T]\cup \{T\}$. The sequence $(\bar \sigma_n)$ defined by
\begin{equation*}
\bar \sigma_n=\left\{\begin{array}{cl} \inf\{t\geq \tau_0;\; \int_{\tau_0}^t \ u^{\tau_0,Y_0}(s)ds
\notin (1-Y_0-L(T-t)+1/n,1-Y_0-1/n)\},& \omega \in M(\tau_0,Y_0) \\ \ &\textnormal{and }\tau_0 \leq T-\frac{T}{n}  \\
     T-\frac{T}{n}               ,& \textnormal{otherwise}.
 \end{array} \right.
\end{equation*}
announces $\bar \sigma$, because $t\mapsto \int_{\tau_0}^t  u^{\tau_0,Y_0}(s)ds $ is continuous and $\tau_0<T$ on $M(\tau_0,Y_0)$. Here, we use the convention
$(a,b)= \emptyset$ for $a\geq b$. Therefore, $\bar \sigma$ is predictable. It, thus, remains to show that $\bar \sigma$ takes values
in $A(\tau_0,Y_0)\cap B(\tau_0,Y_0)$ on the set $M(\tau_0,Y_0)$. We define

\begin{eqnarray*}
\bar \sigma_{\upper}&=&\inf\{t\in [\tau_0,T];\; \int_{\tau_0}^t  u^{\tau_0,Y_0}(s)ds \geq 1-Y_0 \},\\
\bar \sigma_L&=&\inf\{t\in [\tau_0,T];\; \int_{\tau_0}^t  u^{\tau_0,Y_0}(s)ds \leq 1-Y_0-L(T-t) \},
\end{eqnarray*}
(with the usual convention that the infimum of the empty set is $+\infty$). Then $\bar \sigma=\bar\sigma_{\upper} \wedge \bar\sigma_L$ on $M(\tau_0,Y_0)$.
First note that by definition of $\bar \sigma_{\upper}$ and $\bar\sigma_{\nd{L}}$
we obtain on $M(\tau_0,Y_0)$
\begin{eqnarray*}
 \cb{\forall {\epsilon>0}}\quad  \lambda(\{s\in [\tau_0\vee (\bar\sigma_{\upper}-\epsilon), \bar \sigma_{\upper}];\,
 u^{\tau_0,Y_0}(s)>0\})>0, \;\textnormal{ on }\{\bar \sigma_{\upper}\leq T\},\\
\cb{\forall {\epsilon>0}}\quad  \lambda(\{s\in [\tau_0\vee(\bar\sigma_L-\epsilon), \bar \sigma_L];\,
 u^{\tau_0,Y_0}(s)<L\})>0,\;\textnormal{ on }\{\bar \sigma_L\leq T\}.
\end{eqnarray*}
Hence, $\bar \sigma_{\upper} \in B(\tau_0,Y_0)$ on $\{\bar \sigma_{\upper}\leq T\}$ and $\bar \sigma_L \in A(\tau_0,Y_0)$ on $\{\bar \sigma_L\leq T\}$. In particular, we have $\bar \sigma \in
A(\tau_0,Y_0) \cap  B(\tau_0,Y_0)$ on $\{\bar \sigma_{\upper}=\bar \sigma_L\}$. It now suffices to show that $\bar \sigma_{\upper} \in A(\tau_0,Y_0)$
on $\{\bar \sigma_{\upper}< \bar \sigma_L\}$ and $\bar \sigma_L \in B(\tau_0,Y_0)$
on $\{\bar \sigma_L< \bar \sigma_{\upper}\}$. Obviously, $ u^{\tau_0,Y_0}(r)=0<L$ almost everywhere
on $[\nd{\bar\sigma}_{\upper},T]$. However, $\bar\sigma_{\upper}<T$ on $\{\bar \sigma_{\upper}< \bar \sigma_L\}$. Therefore, $\bar \sigma_{\upper} \in A(\tau_0,Y_0)$
on $\{\bar \sigma_{\upper}< \bar \sigma_L\}$. Now, by (\ref{U'}), we get $u^{\tau_0,Y_0}(r)=L>0$ almost everywhere
on $[\sigma_L,T]$. This implies $\bar \sigma_L \in B(\tau_0,Y_0)$
on $\{\bar \sigma_L< \bar \sigma_{\upper}\}$.
\end{proof}

The next lemma relates the left-hand side derivative $D^-_yJ(\tau_0,Y_0)$ to stopping times in $\mathcal{S}_{A(\tau_0,Y_0)}^p$. We recall that $X$ is said to be LCE at a stopping time
$\sigma$, if (\ref{LCE}) holds for every nondecreasing sequence of $[0,T]$-valued stopping times $(\sigma_n)_{n\in \mathbb{N}}$
with  limit $\sigma$.
\begin{lem}\label{lem:A}
 Suppose $X$ satisfies the standing assumptions, $(\tau_0,Y_0)$ are as above, and $\sigma\in \mathcal{S}_{A(\tau_0,Y_0)}^p$. Then, there is a sequence
of stopping times $(\rho_n)$, taking values in $[\tau_0,T]$, which nondecreasingly converges to $\sigma$ and such that
$$
-D^-_yJ(\tau_0,Y_0)\geq  E[X(\rho_n)|\mathcal{F}_{\tau_0}]
$$
on $M(\tau_0,Y_0)$ for every $n\in \mathbb{N}$.
Moreover,
$$
-E[D^-_yJ(\tau_0,Y_0){\bf 1}_{\cb{M(\tau_0,Y_0) }}] \geq  E[X(\sigma) {\bf 1}_{\cb{M(\tau_0,Y_0)}}],
$$
if $X$ is LCE at $\sigma$.
\end{lem}
\begin{proof}
 As $\sigma$ is predictable, there is a sequence of stopping times $(\tilde \sigma_n)$ which announces $\sigma$.
Then the sequence $(\sigma_n)=(\tilde \sigma_n\vee\tau_0)$ nondecreasingly converges to $\sigma$ and satisfies $\tau_0\leq \sigma_n<\sigma$ on $M(\tau_0,Y_0)\subset \{\tau_0<T\}$. We define
$$
\tilde \rho_n=\inf\{t\geq \sigma_n;\;\int_{\sigma_n}^t (L- u^{\tau_{\nd{0}},Y_0}(s))ds> 0\}\wedge T
$$
and for $h>0$
$$
\tilde \rho_{n,h}=\inf\{t\geq \sigma_n;\;\int_{\sigma_n}^t (L- u^{\tau_{\nd{0}},Y_0}(s))ds\geq h\}\wedge T.
$$
Then, $\tilde \rho_{n,h}$ converges to $\tilde \rho_n$ as $h\downarrow 0$.  As
$$
\int_{\sigma_n}^{\tilde \rho_n} (L- u^{\tau_{\nd{0}},Y_0}(s))ds=0,
$$
on $M(\tau_0,Y_0)$, we can conclude that $ u^{\tau_{\nd{0}},Y_0}(s)=L$  \cb{for almost every $(s,\omega)$ such that $s\in[\sigma_n(\omega),\tilde \rho_n(\omega)]$ and $\omega\in M(\tau_0,Y_0)$}. Now, taking into account that
$\sigma\in A(\tau_0,Y_0)$ and $\sigma_n<\sigma$ on  $M(\tau_0,Y_0)$, we observe that $\sigma_n \leq  \tilde \rho_n\leq \sigma$ on $M(\tau_0,Y_0)$. We now define
$$
\rho_n=\left\{\begin{array}{cl}
     \tilde \rho_n, & \omega \in    M(\tau_0,Y_0) \\ \sigma, & \textnormal{otherwise}
              \end{array}
 \right.
$$
Then, $(\rho_n)$ nondecreasingly converges to $\sigma$. Let
$$
u_{n,h}(t)=(L- u^{\tau_0,Y_0}(t)){\bf 1}_{[\tilde \rho_n,\tilde \rho_{n,h}]}(t).
$$
Then, $ u^{\tau_0,Y_0}+u_{n,h}\in U(\tau_0,Y_0-h)$. Hence,
\begin{eqnarray*}
\frac{J(\tau_0,Y_0-h)-J(\tau_0,Y_0)}{h} \geq \frac{1}{h} E[\int_{\tilde \rho_n}^{\tilde \rho_{n,h}} X(t)(L- u^{\tau_0,Y_0}(t)) dt |\mathcal{F}_{\tau_0}].
\end{eqnarray*}
On the set $M(\tau_0,Y_0)$ we have
\begin{eqnarray}\label{eq:hilf13}
 && \frac{1}{h} E[\int_{\tilde \rho_n}^{\tilde \rho_{n,h}} X(t)(L- u^{\tau_0,Y_0}(t)) dt |\mathcal{F}_{\tau_0}] \nonumber \\
&=& E[X(\rho_n) h^{-1} \int_{\rho_n}^{\tilde \rho_{n,h} } (L-u^{\tau_0,Y_0}(t)) dt |\mathcal{F}_{\tau_0}] \nonumber \\ && +  \frac{1}{h} E[\int_{\rho_n}^{\tilde \rho_{n,h} } (X(t)-X(\rho_n))(L-u^{\tau_0,Y_0}(t)) dt |\mathcal{F}_{\tau_0}]
\end{eqnarray}
and
$$
 \lim_{h\downarrow 0} h^{-1} \int_{\rho_n}^{\tilde \rho_{n,h}} (L- u^{\tau_0,Y_0}(t)) dt=1.
$$
So, the first term on the righthand side of (\ref{eq:hilf13}) converges to $ E[X(\rho_n)|\mathcal{F}_{\tau_0}]$ as $h\downarrow 0$. The second term on the righthand side of (\ref{eq:hilf13}) converges to zero
by right-continuity of $X$. Consequently,
$$
-D^-_yJ(\tau_0,Y_0)\geq E[X(\rho_n)|\mathcal{F}_{\tau_0}]
$$
on  $M(\tau_0,Y_0)$. As $\rho_n= \sigma$ on the complement of $M(\tau_0,Y_0)$, we then obtain
\begin{eqnarray*}
 -E[D^-_yJ(\tau_0,Y_0){\bf 1}_{\cb{M(\tau_0,Y_0) }}] \geq  E[X(\rho_n)]- E[X(\sigma) (1-{\bf 1}_{\cb{M(\tau_0,Y_0) }})]
\end{eqnarray*}
If $X$ is LCE at $\sigma$, the right-hand side converges to
 $E[X(\sigma){\bf 1}_{\cb{M(\tau_0,Y_0) }}]$, which completes the proof.
\end{proof}

The corresponding result for the right-hand side derivative reads as follows.
\begin{lem}\label{lem:B}
 Suppose $X$ satisfies the standing assumptions, $(\tau_0,Y_0)$ are as above, and $\sigma\in \mathcal{S}_{B(\tau_0,Y_0)}^p$. Then, there is a sequence
of stopping times $(\rho_n)$, taking values in $[\tau_0,T]$, which nondecreasingly converges to $\sigma$ and such that
$$
-D^+_yJ(\tau_0,Y_0)\leq  E[X(\rho_n)|\mathcal{F}_{\tau_0}]
$$
on $M(\tau_0,Y_0)$ for every $n\in \mathbb{N}$.
Moreover,
$$
-E[D^+_yJ(\tau_0,Y_0) {\bf 1}_{M(\tau_0,Y_0)}]\leq  E[X(\sigma){\bf 1}_{M(\tau_0,Y_0)}],
$$
 if $X$ is LCE at $\sigma$.
\end{lem}
\begin{proof}
 The proof is similar to the proof of the previous lemma, starting from
\begin{eqnarray*}
\tilde \rho_n&=&\inf\{t\geq \sigma_n;\;\int_{\sigma_n}^t  u^{\tau_{\nd{0}},Y_0}(s)ds> 0\}\wedge T \\
\tilde \rho_{n,h}&=&\inf\{t\geq \sigma_n;\;\int_{\sigma_n}^t  u^{\tau_{\nd{0}},Y_0}(s)ds\geq h\}\wedge T,
\end{eqnarray*}
where the sequence $(\sigma_n)$ again nondecreasingly converges to $\sigma$ and satisfies $\tau_0\leq \sigma_n<\sigma$ on $M(\tau_0,Y_0)$. It has one additional complication, namely that  $ u^{\tau_0,Y_0}- u^{\tau_0,Y_0} {\bf 1}_{[\tilde \rho_n,\tilde \rho_{n,h}]}$
 does in general not belong to $U(\tau_0,Y_0+h)$. As a remedy we fix some arbitrary $m\in \mathbb{N}$ and assume $h<1/m$. We introduce the set
$$
M_m(\tau_0,Y_0):=\{L(T-\tau_0)>1-Y_0\geq 1/m\} \in \mathcal{F}_{\tau_0}
$$
and the stopping times
\begin{eqnarray*}
\cb{\bar\sigma_{U,h}}&:=& \inf\{t\geq \tau_0;\;Y_0+\int_{\tau_0}^t  u^{\tau_{\nd{0}},Y_0}(s)ds\geq 1-h\}\wedge T\\
\cb{\bar\sigma_{U}}&:=& \inf\{t\geq \tau_0;\;Y_0+\int_{\tau_0}^t  u^{\tau_{\nd{0}},Y_0}(s)ds\geq 1\}\wedge T
\end{eqnarray*}
Let
$$
u_{n,h}(t)= u^{\tau_0,Y_0}(t){\bf 1}_{[\tilde \rho_n\wedge \cb{\bar\sigma_{U,h}} ,\tilde \rho_{n,h}\wedge \cb{\bar\sigma_{U}}]}(t) {\bf 1}_{M_m(\tau_0,Y_0)}.
$$
As $\int_{\tau_0}^T u^{\tau_0,Y_0}(r)dr=1-Y_0$ on $M_m(\tau_0,Y_0)\subset M(\tau_0,Y_0)$ by Proposition \ref{prop:optimalcontrol}, we conclude
that
$$
\int_{\tau_0}^T (u^{\tau_0,Y_0}(r)- u_{n,h}(r)) dr=1-Y_0-h
$$
on $M_m(\tau_0,Y_0)$.  This implies that $ u^{\tau_0,Y_0}-u_{n,h}\in U(\tau_0,Y_0+h{\bf 1}_{M_m(\tau_0,Y_0)})$.
Consequently, on
$M_m(\tau_0,Y_0)$,
\begin{eqnarray}\label{eq:hilf0002}
&& -\frac{J(\tau_0,Y_0+h)-J(\tau_0,Y_0)}{h} \leq \frac{1}{h} E[\int_{\tau_0}^T  u_{n,h}(t)X(t) dt|\mathcal{F}_{\tau_0}] \nonumber \\ &\leq&
\frac{1}{h}  E[\int_{\tilde \rho_n}^{\tilde \rho_{n,h}}  u^{\tau_0,Y_0}(t)X(t) dt|\mathcal{F}_{\tau_0}] + E[\sup_{r\in [0,T] } X(r) {\bf 1}_{\{\cb{\bar\sigma_{U,h}}<\tilde \rho_n\}} |\mathcal{F}_{\tau_0}]
\end{eqnarray}
Note that by the definition of $\tilde \rho_n$, we have
$$
1-Y_0-\int_{\tau_0}^{\tilde \rho_n} u^{\tau_{\nd{0}},Y_0}(s)ds= 1-Y_0-\int_{\tau_0}^{\sigma_n}  u^{\tau_{\nd{0}},Y_0}(s)ds
$$
on $M_m(\tau_0,Y_0)$. Moreover, this expression is strictly positive on $M_m(\tau_0,Y_0)$, because $\sigma\in \mathcal{S}_{B(\tau_0,Y_0)}^p$
and $\sigma_n<\sigma$. Hence, $\tilde \rho_n\leq  \cb{\bar\sigma_{U,h}}$ for $h$ sufficiently small (depending on $\omega$). This shows that the second term in
(\ref{eq:hilf0002}) tends to zero as $h$ goes to zero and that
$$
 \lim_{h\downarrow 0} h^{-1} \int_{\tilde \rho_n}^{\tilde \rho_{n,h}} u^{\tau_0,Y_0}(t) dt=1.
$$
Now the same argument as in (\ref{eq:hilf13}) can be applied to the first term in (\ref{eq:hilf0002}). We hence conclude that
$$
-D^+_yJ(\tau_0,Y_0)\leq  E[X(\tilde \rho_n)|\mathcal{F}_{\tau_0}]
$$
on $M_m(\tau_0,Y_0)$ for every $m\in \mathbb{N}$, and thus, on $M(\tau_0,Y_0)=\cup_{m\in \mathbb{N}} M_m(\tau_0,Y_0)$. Defining
$$
\rho_n=\left\{\begin{array}{cl}
     \tilde \rho_n, & \omega \in    M(\tau_0,Y_0) \\ \sigma, & \textnormal{otherwise,}
              \end{array}
 \right.
$$
the rest of the proof is identical to the one of the previous lemma.
\end{proof}

As a consequence of the previous three lemmas we get the following criterion for the left-hand side derivative and the right-hand side
derivative of $J$ to coincide.
\begin{prop}\label{prop:derivative}
 Suppose $X$ satisfies the standing assumptions and  $(\tau_0,Y_0)$ are as above. If $X$ is LCE at $\bar\sigma$, defined in (\ref{eq:barsigma}),
then
$$
D^-_yJ(\tau_0,Y_0)=D^+_yJ(\tau_0,Y_0)
$$
on $\{1>Y_0\neq 1-L(T-\tau_0)\}$, $P$-almost surely. Moreover, $D^-_yJ(\tau_0,Y_0)=0$ on $\{Y_0\leq 1-L(T-\tau_0)\}$.
\end{prop}
\begin{proof}
 The previous three lemmas imply that
$$
E[D^+_yJ(\tau_0,Y_0) {\bf 1}_{M(\tau_0,Y_0)}]\geq - E[X(\bar \sigma){\bf 1}_{M(\tau_0,Y_0)}] \geq E[D^-_yJ(\tau_0,Y_0) {\bf 1}_{M(\tau_0,Y_0)}].
$$
As $D^-_yJ(\tau_0,Y_0) \geq D^+_yJ(\tau_0,Y_0)$ by concavity, we conclude that
$$
D^-_yJ(\tau_0,Y_0)=D^+_yJ(\tau_0,Y_0)
$$
on $M(\tau_0,Y_0)$. On the set $\{Y_0 \leq 1- L(T-\tau_0)\}$, we have
$$
J(\tau_0,Y_0)=E[\int_{\tau_0}^T LX(s)ds|\mathcal{F}_{\tau_0}].
$$
Hence, $D^-_yJ(\tau_0,Y_0)=0$ on  $\{Y_0 \leq 1- L(T-\tau_0)\}$ and  $D^+_yJ(\tau_0,Y_0)=0$ on  $\{Y_0 < 1- L(T-\tau_0)\}$.
\end{proof}

We are now in the position to give the proof of Theorem \ref{thm:smooth}, which at the same time finishes the proof of Theorem \ref{thm:main_new}.
\begin{proof}[Proof of Theorem \ref{thm:smooth}]
The case $t=T$ is trivial, because $J(T,y)=0$ for every $y\in(-\infty,1)$.  We now fix $t\in [0,T)$. Then, for every $\omega \in \bar \Omega$ (which is the set of full measure introduced in
Proposition \ref{prop:goodversion}), the mapping
$$
y\mapsto D^-_yJ(t,\omega,y)
$$
is nonincreasing and left-continuous on $(1-L(T-t),1)$ by concavity. Hence, there is a countable family of $\mathcal{F}_t$-measurable random variables $(Y_n)_{n \in \mathbb{N}}$ with values in   $[1-L(T-t),1]$
such that the jumps of $y\mapsto D^-_yJ(t,y)$, restricted to $(1-L(T-t),1)$,  are included in $(Y_n)_{n \in \mathbb{N}}$, More precisely,
\begin{eqnarray*}
D(\omega)&:=&\{y\in (1-L(T-t),1);\;  D^-_yJ(t,\omega,y)\neq \lim_{\eta \downarrow \cb{y}}  D^-_yJ(t,\omega,\eta)\}\\ &=&
\{y\in (1-L(T-t),1);\;  D^-_yJ(t,\omega,y)\neq   D^+_yJ(t,\omega,y)\} \\&\subset& \{ Y_1(\omega),Y_2(\omega),\ldots\}.
\end{eqnarray*}
Here the first identity follows again by concavity. By the previous proposition there is a set $\Omega_t\subset \bar \Omega$ of full $P$-measure such that, for every $n\in \mathbb{N}$
and $\omega \in \Omega_t$,
$$
Y_n(\omega)\in (1-L(T-t),1)\;\Longrightarrow \; D^-_yJ(t,\omega,Y_n(\omega))= D^+_yJ(t,\omega,Y_n(\omega))
$$
This implies that $D(\omega)=\emptyset$ for $\omega\in\Omega_t$ and, hence,
$$
y\mapsto J(t,\omega,y)
$$
is continuously differentiable on $(1-L(T-t),1)$ for $\omega\in \Omega_t$.

\cb{In order to prove the second assertion, we introduce the set
$$
C:=\left\{(t,\omega,y)\in[0,T]\times \Omega\times(-\infty,1);\; D^+_y J(t,\omega,y)=D^-_yJ(t,\omega,y)\right\}\in \mathcal{B}_{[0,T]}\otimes \mathcal{F}\otimes \mathcal{B}_{(-\infty,1)}.
$$
We fix $t\in [0,T]$. For $y\in (1-L(T-t),1)$  and $s\in [t,T-\frac{1-y}{L}]$ we consider the cuts
$$
C_{y}=\{ (r,\omega) \in [t,T-\frac{1-y}{L}]\times \Omega;\; (r,\omega,y)\in C\},\quad C_{(s,y)}=\{ \omega \in \Omega;\; (s,\omega,y)\in C\}.
$$
Then, for every $y\in (1-L(T-t),1)$  and $s\in [t,T-\frac{1-y}{L})$
$$
\Omega_s\subset\left\{\omega \in \Omega;\; \forall  \eta\in( 1-L(T-s),1) \quad D^+_y J(s,\omega,\eta)=D^-_yJ(s,\omega,\eta) \right\}\subset C_{(s,y)},
$$
where $\Omega_s$ was constructed in the first part of the proof. Consequently, $P(C_{(s,y)})=1$ for every $y\in (1-L(T-t),1)$  and $s\in [t,T-\frac{1-y}{L})$. An application of
Fubini's theorem then yields for every $y\in (1-L(T-t),1)$
\begin{eqnarray*}
 \lambda_{[t,T-\frac{1-y}{L}]}\otimes P(C_y)=\int_t^{T-\frac{1-y}{L}} P(C_{(s,y)}) ds =T-\frac{1-y}{L}-t,
\end{eqnarray*}
 which finishes the proof.}
\end{proof}

The following theorem relates the derivative of $J$ explicitly to optimal stopping problems, if $X$ is LCE. It can be considered as the main result of this section.

\begin{thm}\label{thm:smooth2}
 Suppose $X$ satisfies the standing assumptions and it is LCE. Then, for every $[0,T]$-valued stopping time $\tau_0$ and every $\mathcal{F}_{\tau_0}$-measurable, $(-\infty,1]$-valued random variable
$Y_0$ the following holds:\\[0.1cm]
(i) On the set $\{L(T-\tau_0)>1-Y_0>0\}$
$$
-D^-_yJ(\tau_0,Y_0)=-D^+_yJ(\tau_0,Y_0)=\esssup_{\sigma \in \mathcal{S}_{A(\tau_0,Y_0)}^p} E[X(\sigma)|\mathcal{F}_{\tau_0}]=
\essinf_{\rho \in \mathcal{S}_{B(\tau_0,Y_0)}^p} E[X(\rho)|\mathcal{F}_{\tau_0}].
$$
 Moreover, every stopping time from the nonempty set $\mathcal{S}_{A(\tau_0,Y_0)}^p\cap\mathcal{S}_{B(\tau_0,Y_0)}^p$
is optimal for both optimal stopping problems.\\[0.1cm]
(ii) On the set $\{Y_0<1-L(T-\tau_0)\}$
$$
-D^-_yJ(\tau_0,Y_0)=-D^+_yJ(\tau_0,Y_0)=0.
$$
(iii) On the set $\{Y_0=1-L(T-\tau_0)\}\cap\{\tau_0<T\}$
$$
-D^-_yJ(\tau_0,Y_0)=0,\quad -D^+_yJ(\tau_0,Y_0)=\essinf_{\sigma \in \mathcal{S}_{\tau_0}} E[X(\sigma)|\mathcal{F}_{\tau_0}],
$$
where $\mathcal{S}_{\tau_0}$ denotes the set of stopping times with values in $[\tau_0,T]$. \\[0.1cm]
(iv) On the set $\{Y_0=1\}\cap\{\tau_0<T\}$
$$
-D^-_yJ(\tau_0,Y_0)=\esssup_{\sigma \in \mathcal{S}_{\tau_0}} E[X(\sigma)|\mathcal{F}_{\tau_0}].
$$
\end{thm}
\begin{rem}
 A related representation for the marginal value of a discrete time multiple stopping problem is derived in Theorem 2.2 of \citet{Be1}.
We also note that differentiability of the  value process for some class of finite fuel problems related to the monotone follower problem
 can be shown by  expressing the derivative explicitly in terms of (standard) optimal stopping problems,
see e.g. \citet{Ka,KS86}.
\end{rem}

Before we provide the proof, we note that the two stopping problems in Theorem \ref{thm:smooth2}, (i), make the intuition at the beginning of this section rigorous.
 The marginal value can be calculated by adding some marginal volume at the
best time where exercise is still possible. It can also be calculated by removing some marginal volume at the cheapest time, where this is possible.
The interesting aspect is that here `best time' and `cheapest time'
refer to predictable stopping times only. The next example shows that this restriction is essential.
\begin{exmp}
 Suppose $\xi$ is a binary trial with $P(\{\xi=1\})=P(\{\xi=-1\})$. Define
$$
X(t)=1+\xi(2-t){\bf 1}_{[1,3]}(t), \quad t\in [0,3],
$$
and consider the filtration $(\mathcal{F}_t)_{t\in[0,3] }$ generated by $X$. Then $X$ satisfies the standing assumptions on the time horizon $[0,3]$ and is LCE as the sum of the martingale $1+\xi{\bf 1}_{[1,3]}(t)$ and the continuous process
$\xi(1-t){\bf 1}_{[1,3]}(t)$. We assume $L=1$. For $t\in[0,1]$ and $y\in[0,1]$ it is then straightforward to see that
$$
J(t,y)=2-\frac{(1+y)^2}{2}
$$
and that
$$
 u^{t,y}(r)={\bf 1}_{[1,2-y]}(r){\bf 1}_{\{X(1)=2\}}+ {\bf 1}_{[2+y,3]}(r){\bf 1}_{\{X(1)=0\}}
$$
is optimal. In particular,
$$
A(0,1/2)=\left\{\begin{array}{cl}
      (0,5/2], & X(1)=0;\\ \
(0,1]\cup [3/2,3], & X(1)=2.
                \end{array}
 \right.
$$
Define the (non-predictable) stopping time
$$
\tau=\inf\{r\geq 0;\; X(r)\geq 3/2\}= {\bf 1}_{\{X(1)=2\}} + 5/2 \,{\bf 1}_{\{X(1)=0\} },
$$
which takes values in $A(0,1/2)$.
Then,
$$
E[X(\tau)]=7/4>3/2=-D^-_yJ(0,1/2)=\sup_{\sigma \in \mathcal{S}_{A(0,1/2)}^p} E[X(\sigma)],
$$
where the last identity is due to Theorem \ref{thm:smooth2}. However, in view of Lemma \ref{lem:barsigma} and Theorem \ref{thm:smooth2}, an optimal stopping time in $\mathcal{S}_{A(0,1/2)}^p$
is given by
\begin{eqnarray*}
\bar \sigma&=&\inf\{r\geq 0;\; \int_0^r u^{0,1/2}(s)ds \notin (t-5/2,1/2)\}= 3/2\;{\bf 1}_{\{X(1)=2\}} + 5/2 \,{\bf 1}_{\{X(1)=0\} }\\&=&\inf\{r\geq 0;\; X(r)=3/2\}.
\end{eqnarray*}
This shows that
the restriction to predictable stopping times cannot be avoided in the optimal stopping characterization of the $y$-derivative of $J$.
\end{exmp}

We now give the proof of Theorem \ref{thm:smooth2}.

\begin{proof}[Proof of Theorem \ref{thm:smooth2}]
 (i) Choose some stopping time $\sigma \in \mathcal{S}_{A(\tau_0,Y_0)}^p$.
 By Lemma \ref{lem:A}, there is a sequence
of stopping times $(\rho_n)$ with values in $[\tau_0,T]$, which nondecreasingly converges to $\sigma$ and satisfies
$$
-D^-_yJ(\tau_0,Y_0)\geq  E[X(\rho_n)|\mathcal{F}_{\tau_0}]
$$
on $M(\tau_0,Y_0)$.
By the integrability property of $X$ in (\ref{ass:square}) and left-continuity in expectation one easily obtains
$$
\lim_{n\rightarrow \infty} E[X(\rho_n)|\mathcal{F}_{\tau_0}]=E[X(\sigma)|\mathcal{F}_{\tau_0}],
$$
because $\rho_n\geq \tau_0$. Hence,
$$
-D^-_yJ(\tau_0,Y_0)\geq  E[X(\sigma)|\mathcal{F}_{\tau_0}]
$$
on $M(\tau_0,Y_0)$.
Analogously we obtain
$$
-D^+_yJ(\tau_0,Y_0)\leq  E[X(\rho)|\mathcal{F}_{\tau_0}]
$$
on $M(\tau_0,Y_0)$
for $\rho \in \mathcal{S}_{B(\tau_0,Y_0)}^p$ making use of Lemma \ref{lem:B}. Hence, choosing $\tilde \sigma$ from the nonempty set  $\mathcal{S}_{A(\tau_0,Y_0)}^p\cap \mathcal{S}_{B(\tau_0,Y_0)}^p$ (by Lemma
\ref{lem:barsigma}), we get
\begin{eqnarray*}
 -D^-_yJ(\tau_0,Y_0)&\geq&  \esssup_{\sigma \in \mathcal{S}_{A(\tau_0,Y_0)}^p} E[X(\sigma)|\mathcal{F}_{\tau_0}] \geq E[X(\tilde \sigma)|\mathcal{F}_{\tau_0}]
 \\ &\geq & \essinf_{\rho \in \mathcal{S}_{B(\tau_0,Y_0)}^p} E[X(\rho)|\mathcal{F}_{\tau_0}] \geq -D^+_yJ(\tau_0,Y_0)
\end{eqnarray*}
on $M(\tau_0,Y_0)$. As $D^-_yJ(\tau_0,Y_0)\geq D^+_yJ(\tau_0,Y_0)$ by concavity, the assertion follows.
\\[0.1cm]
(iii) Define $Y_n:=Y_0+(1-Y_0)/n$. Denote by $u^{\tau_0,Y_n}$ an optimal control for $\bar J(\tau_0,Y_n)$. Then,
$Y_n>Y_0$ on $\{Y_0=1-L(T-\tau_0)\}\cap\{\tau_0<T\}$. Hence, on this set,
\begin{eqnarray*}
 -D^+_yJ(\tau_0,Y_0)&=&\lim_{n\rightarrow \infty} \frac{-J(\tau_0,Y_n)+J(\tau_0,1-L(T-\tau_0))}{Y_n-Y_0}
\\ &=& \lim_{n\rightarrow \infty} \frac{E\left[\left. \int_{\tau_0}^T (L-u^{\tau_0,Y_n}(s))X(s) \right|\mathcal{F}_{\tau_0}\right]}{(1-Y_0)/n}.
\end{eqnarray*}
  Denoting by $Y_*(s)$ an RCLL version of the submartingale
$s\mapsto \essinf_{\sigma \in \mathcal{S}_s} E[X(\sigma)|\mathcal{F}_s],
$
we, hence, obtain thanks to Corollary \ref{cor:submartingale},
\begin{eqnarray*}
 -D^+_yJ(\tau_0,Y_0)&\geq& \limsup_{n\rightarrow \infty} \frac{E\left[\left. \int_{\tau_0}^T (L-u^{\tau_0,Y_n}(s))Y_*(s) \right|\mathcal{F}_{\tau_0}\right]}{(1-Y_0)/n} \\
&\geq& \limsup_{n\rightarrow \infty} \frac{E\left[\left. \int_{\tau_0}^T (L-L {\bf 1}_{[\tau_0\vee (T-(1-Y_n)/L),T]}(s))Y_*(s) \right|\mathcal{F}_{\tau_0}\right]}{(1-Y_0)/n}
\\ &=& \limsup_{n\rightarrow \infty} \frac{E\left[\left. \int_{\tau_0}^{\tau_0+(1-Y_0)/(Ln)} Y_*(s) \right|\mathcal{F}_{\tau_0}\right]}{(1-Y_0)/(Ln)}
\\ &=& Y_*(\tau_0)=\essinf_{\sigma \in \mathcal{S}_{\tau_0}} E[X(\sigma)|\mathcal{F}_{\tau_0}].
\end{eqnarray*}
For the reverse inequality fix some  some arbitrary stopping time $\sigma$ with values in $[\tau_0, T]$.
Define $\sigma_m=\sigma\wedge(T-(T-\tau_0)/m)\geq \tau_0$.  We observe that, for $n\geq m$,
\begin{eqnarray*}
u_{n,m}&:=&L-L{\bf 1}_{[\sigma_m, \sigma_m+(T-\tau_0)/n]}{\bf 1}_{\{Y_0=1-L(T-\tau_0)\}\cap\{\tau_0<T\}} \\ &\in& U(\tau_0, L(T-\tau_0)-{\bf 1}_{\{Y_0=1-L(T-\tau_0)\}\cap\{\tau_0<T\}} L(T-\tau_0)/n).
\end{eqnarray*}
Thus, on $\{Y_0=1-L(T-\tau_0)\}\cap\{\tau_0<T\}$,
\begin{eqnarray*}
  -D^+_yJ(\tau_0,Y_0)&=&\lim_{n\rightarrow \infty} \frac{-J(\tau_0,L(T-\tau_0)(1+1/n))+J(\tau_0,1-L(T-\tau_0))}{L(T-\tau_0)/n}
\\ &\leq & \liminf_{n\rightarrow \infty} \frac{E\left[\left. \int_{\tau_0}^T (L-u_{n,m}(s))X(s) \right|\mathcal{F}_{\tau_0}\right]}{L(T-\tau_0)/n}
\\ &=& \liminf_{n\rightarrow \infty} \frac{E\left[\left. \int_{\sigma_m}^{\sigma_m+(T-\tau_0)/n} X(s) \right|\mathcal{F}_{\tau_0}\right]}{(T-\tau_0)/n}=E[X(\sigma_m)|\mathcal{F}_{\tau_0}]
\end{eqnarray*}
by right-continuity of $X$. Finally, by left-continuity in expectation we obtain
$$
-D^+_yJ(\tau_0,Y_0)\leq \lim_{m\rightarrow \infty} E[X(\sigma_m)|\mathcal{F}_{\tau_0}]=E[X(\sigma)|\mathcal{F}_{\tau_0}].
$$
As $\sigma$ was arbitrary this concludes the proof for the right-hand side derivative. The statement for the left-hand side derivative is included in Proposition \ref{prop:derivative}.
\\[0.1cm] (iv) is similar to (iii), and (ii)
is a direct consequence of Proposition \ref{prop:derivative}.
\end{proof}

\section{Duality}

In this final section of the paper we discuss a dual minimization problem to $\bar J(0,0)$. Recall that admissible controls in $U(0,0)$ are subject to three constraints.
A local constraint requires that $u$ takes values in $[0,L]$, a global one imposes that the total volume spent by the investor $\int_0^T u(s)ds$ is bounded by one, and the third one
is the adaptedness condition. In Section 7 of [BD] we relaxed the adaptedness constraint and came up with a continuous time version of an information relaxation
dual. This kind of dual is well studied for discrete time stochastic control problems, see e.g.  \citet{BSS}. We now relax the global constraint and re-inforce it by a more classical Lagrange multiplier approach.
It turns out that the Lagrange multiplier can be calculated explicitly in term of the derivative of $J$.
This approach leads to the following result:

\begin{thm}\label{thm:dual}
Suppose $LT>1$, and that $X$ satisfies the standing assumptions and is LCE.  Denote by $\mathcal{M}$ the set of \cb{RCLL }martingales on $[0,T]$. Then
$$
\bar J(0,0)=\inf_{M\in \mathcal{M}} \left(E\left[\int_0^T L (X(t)-M(t))_+\,dt\right] +E[M(0)] \right).
$$
Moreover, an optimal martingale is given by \cb{the (unique up to indistinguishability) RCLL and adapted modification of}
$$
\bar M(t)=\left\{\begin{array}{cl} -\cb{\frac{\partial }{\partial y}}J\left(t,\int_0^t u^{0,0}(s)ds\right), & t< \bar\sigma \\
Y^*(\bar \sigma)+M^*(t)-M^*(\bar \sigma), & t\geq\bar\sigma=\bar\sigma_{\upper}  \\
          Y_*(\bar \sigma)+M_*(t)-M_*(\bar \sigma), & t\geq \bar\sigma=\bar\sigma_L,
  \end{array}
\right.
$$
where $\bar \sigma=\bar \sigma_{\upper}\wedge\bar \sigma_L$,
\begin{eqnarray*}
\bar \sigma_{\upper}&=&\inf\{t\geq 0;\; \int_{0}^t  u^{0,0}(s)ds \geq 1 \},\\
\bar \sigma_L&=&\inf\{t\geq 0;\; \int_{0}^t  u^{0,0}(s)ds \leq 1-L(T-t) \},
\end{eqnarray*}
and $u^{0,0}$ is an optimal control for $\bar J(0,0)$ satisfying $\int_0^T u^{0,0}(s)ds=1$. Moreoever,
 $M^*$ and $M_*$ are the martingale parts of the Doob-Meyer decompositions of the RCLL supermartingale
$Y^*(s)=\esssup_{\sigma \in \mathcal{S}_s} E[X(\sigma)|\mathcal{F}_s]$ and the RCLL submartingale
$Y_*(s)=\essinf_{\sigma \in \mathcal{S}_s} E[X(\sigma)|\mathcal{F}_s]$.
\end{thm}
We prepare the proof of this theorem with the following lemma.
\begin{lem}\label{lem:dual}
 Under the assumptions and with the notations of Theorem \ref{thm:dual} we have:
\begin{eqnarray}
 E[X(\bar \sigma)|\mathcal{F}_t]&=&-\cb{\frac{\partial }{\partial y}}J\left(t,\int_0^t u^{0,0}(s)ds\right),\quad t<\bar \sigma, \label{dual1}\\
X(\bar \sigma)&=&\left\{\begin{array}{cl} Y^*(\bar \sigma_{\upper}), & \bar\sigma_{\upper}\leq \bar\sigma_L
\\ Y_*(\bar \sigma_L), & \bar\sigma_L< \bar\sigma_{\upper}
 \end{array}\right.
 \label{dual2}
\end{eqnarray}
\end{lem}
\begin{rem}\label{rem:martingale}
 By Lemma \ref{lem:dual}, the process $\bar M$ in Theorem \ref{thm:dual} can be expressed as
$$
\bar M(t)=\left\{\begin{array}{cl} E[X(\bar \sigma)|\mathcal{F}_t], & t\leq \bar\sigma \\
X(\bar \sigma)+M^*(t)-M^*(\bar \sigma), & t>\bar\sigma=\bar\sigma_{\upper}  \\
          X(\bar \sigma)+M_*(t)-M_*(\bar \sigma), & t> \bar\sigma=\bar\sigma_L.
  \end{array}
\right.
$$
In particular, this shows that $\bar M$ is a martingale \cb{and, thus, has an adapted RCLL modification, which we still denote by $\bar M$}.
\end{rem}

\begin{proof}
First note that by Proposition 3.2 in [BD] there is an optimal control $u^{0,0}$ satisfying $\int_0^T u^{0,0}(s)ds=1$. Moreover,
$u^{0,0}$ fulfills (\ref{U'}) with $\tau_0=0$ and $Y_0=0$ by the same proposition.
We first prove (\ref{dual1}). By Theorem \ref{thm:smooth2}, (i), and Lemma \ref{lem:barsigma} we observe that  for $t<\bar\sigma$
\begin{equation}\label{hilf0004}
-\cb{\frac{\partial }{\partial y}}J\left(t,\int_0^t u^{0,0}(s)ds\right)= E[X(\bar \sigma(t))|\mathcal{F}_t]
\end{equation}
where
\begin{eqnarray*}
\bar \sigma(t)&=&\inf\{s\geq t;\; \int_{t}^s \bar u^{0,0}(r)dr
\notin (1-\int_0^t u^{0,0}(r)dr-L(T-t),1-\int_0^t u^{0,0}(r)dr)\} \\
&=&
\inf\{s\geq t;\; \int_{0}^s \bar u^{0,0}(r)dr
\notin (1-L(T-t),1)\}
\end{eqnarray*}
Here we used that ${\bf 1}_{[t,T]} u^{0,0}$ is optimal for $\bar J(t,\int_0^t  u^{0,0}(r)dr)$ by the dynamic programming principle in Proposition 3.3 of [BD].
As $\bar\sigma=\bar \sigma(0)$, we conclude that $\bar \sigma(t)=\bar\sigma$ for $t<\bar\sigma$. Hence, (\ref{hilf0004}) implies (\ref{dual1}).

We now turn to the proof of (\ref{dual2}). Suppose that $\sigma$ is any predictable stopping time with values in $(\bar\sigma_L,T]\cup\{T\}$
and denote an announcing sequence by $(\sigma_n)$. We will first show that
\begin{equation}\label{hilf0006}
 E[X(\bar \sigma){\bf 1}_{\{\bar\sigma_L<\bar \sigma_{\upper}\}}]\leq E[X(\sigma){\bf 1}_{\{\bar\sigma_L<\bar \sigma_{\upper}\}}].
\end{equation}
To this end we define
$$
\tilde \sigma=\left\{\begin{array}{cl} \sigma,& \bar\sigma_L<\bar \sigma_{\upper} \\ \bar\sigma_{\upper}, & \bar\sigma_{\upper}\leq \bar\sigma_L.\end{array} \right.
$$
The stopping time $\cb{\tilde\sigma}$ is predictable, since it is announced by the sequence
$$
\tilde \sigma_n=\left\{\begin{array}{cl} \sigma_n \vee \bar\sigma_{L} ,& \bar\sigma_{L}<\bar \sigma_{\upper,n} \\ \bar\sigma_{\upper,n}, & \bar\sigma_{\upper,n}\leq \bar\sigma_{L},\end{array} \right\}\wedge(T-1/n)
$$
where the sequence
\begin{eqnarray*}
 \bar \sigma_{\upper,n}&=&\inf\{t\in [0,T];\; \int_{0}^t  u^{0,0}(s)ds \geq 1-1/n \}
\end{eqnarray*}
announces $\bar\sigma_{\upper}$. Now note that $\tilde \sigma=\bar \sigma \in B(0,0)$ on $\bar\sigma_{\upper} \leq \bar\sigma_L$ by Lemma \ref{lem:barsigma}. Moreover,
$(\bar \sigma_L,T]\subset B(0,0)$ on $\bar\sigma_L<\bar\sigma_{\upper}$ by (\ref{U'}). Consequently, $\tilde \sigma$ belongs to $\mathcal{S}_{B(0,0)}^p$. Hence,
Theorem \ref{thm:smooth2}, (i), and Lemma \ref{lem:barsigma} yield
$$
E[X(\bar\sigma)]\leq E[X(\tilde \sigma)],
$$
which in turn implies (\ref{hilf0006}) by the definition of $\tilde \sigma$. In a next step we fix an optimal stopping time $\sigma^*$ for the
optimal stopping problem
$Y_*(\cb{\bar\sigma_L}\wedge T)$, which exists, because $X$ is right-continuous and LCE, see e.g. \citet{ElK}. Then, for every $k\in \mathbb{N}$, the stopping time
$\sigma^*_k=(\sigma^*+1/k)\wedge T$ is predictable (with announcing sequence $((\sigma^*+1/k-1/n)\wedge (T-1/n))_{n\geq k}$) and takes values in $(\bar\sigma_L,T]\cup\{T\}$. Thus,
by (\ref{hilf0006}), we have
$$
E[X(\bar \sigma){\bf 1}_{\{\bar\sigma_L<\bar \sigma_{\upper}\}}]\leq E[X(\sigma^*_k){\bf 1}_{\{\bar\sigma_L<\bar \sigma_{\upper}\}}]
$$
Passing to the limit we obtain by right-continuity of $X$ and optimality of $\sigma^*$,
$$
E[X(\bar \sigma){\bf 1}_{\{\bar\sigma_L<\bar \sigma_{\upper}\}}]\leq E[X(\sigma^*){\bf 1}_{\{\bar\sigma_L<\bar \sigma_{\upper}\}}]=E[Y_*(\bar \sigma_L){\bf 1}_{\{\bar\sigma_L<\bar \sigma_{\upper}\}}].
$$
As obviously,
$$
X(\bar\sigma_L)\geq Y_*(\bar \sigma_L)
$$
on $\{\bar\sigma_L<\bar \sigma_{\upper}\}$, we finally arrive at
$$
X(\bar\sigma_L)= Y_*(\bar \sigma_L)
$$
on $\{\bar\sigma_L<\bar \sigma_{\upper}\}$. The proof of (\ref{dual2}) in the case $\bar\sigma_{\upper}< \bar\sigma_L$ is analogously, while the case $\bar\sigma_L=\bar\sigma_{\upper}$ is trivial, since this implies
$\bar\sigma_{\upper}=T$.
\end{proof}

\begin{proof}[Proof of Theorem \ref{thm:dual}]
We introduce the set
$
{\bf U}
$
consisting of all adapted, $[0,L]$-valued processes. Hence, $u\in U(0,0)$, if and only if $u\in {\bf U}$ and satisfies the global constraint $\int_0^T u(s)ds\leq 1$.
We then define the Lagrangian
$$
\mathcal{L}(u,\Lambda)=E\left[\int_0^T u(s)X(s)ds-\Lambda\left(\int_0^T u(s)ds-1 \right) \right]
$$
for $u\in {\bf U}$ and $\mathcal{F}_T$-measurable integrable random variables $\Lambda$. Put differently, we relax the global constraint and try to enforce it by an appropriate choice of
the Lagrange multiplier $\Lambda$. Apparently we have for any $\Lambda$
$$
\mathcal{L}(u^{0,0},\Lambda)=E\left[\int_0^T u^{0,0}(s)X(s)ds\right]=\bar J(0,0),
$$
because $u^{0,0}$ is optimal and satisfies $\int_0^T u^{0,0}(s)ds= 1$. Thus, we obtain for every $\mathcal{F}_T$-measurable integrable random variable $\Lambda$
\begin{eqnarray*}
 \bar J(0,0) &\leq& \sup_{u\in {\bf U}}  \mathcal{L}(u,\Lambda)= \sup_{u\in {\bf U}} E\left[\int_0^T u(s)(X(s)-E[\Lambda|\mathcal{F}_s]) ds\right]+E[\Lambda]
 \\
&=&  E\left[\int_0^T L(X(s)-E[\Lambda|\mathcal{F}_s])_+ ds\right]+E[\Lambda].
\end{eqnarray*}
Consequently, we get for every RCLL martingale $M$ with the choice $\Lambda=M(T)$,
$$
J(0,0)\leq E\left[\int_0^T L(X(s)-M(s))_+ ds\right]+E[M(0)].
$$

In order to finish the proof it is, thus, sufficient to show that
$$
\bar J(0,0)=E\left[\int_0^{T}  L(X(t)- \bar M(t))_+ dt\right]+E[\bar M(0)],
$$
 where $\bar M$ is indeed a martingale thanks to Remark \ref{rem:martingale}.
To this end recall that the good version of the value process constructed in Proposition \ref{prop:goodversion} is denoted by $J$. By the characterization of optimal controls in Theorem \ref{thm:main_old}, (ii),
by the martingale property of $\bar M$, and by (\ref{U'})
we obtain
\begin{eqnarray*}
\bar J(0,0)&=&J(0,0)=E[\int_0^T u^{0,0}(t)X(t)dt]\\&=&E\left[\int_0^T  \left( L(X(t)+ D^-_yJ(t,\int_0^t u^{0,0}(s)ds))_+ - D^-_yJ(t,\int_0^t u^{0,0}(s)ds)u^{0,0}(t) \right) dt\right] \\
&=& E\left[\int_0^{\bar\sigma}  \left( L(X(t)- \bar M(t))_+ + E[\bar M(\bar \sigma)|\mathcal{F}_t] u^{0,0}(t) \right) dt\right]
\\ &&+ E\left[{\bf 1}_{\{\bar\sigma_{\upper}\leq \bar\sigma_L\}} \int_{\bar\sigma}^T  \left( L(X(t)+ D^-_yJ(t,1))_+ - D^-_yJ(t,1)u^{0,0}(t) \right) dt\right] \\
&& + E\left[{\bf 1}_{\{\bar\sigma_{\upper}> \bar\sigma_L\}} \int_{\bar\sigma}^T  \left( L(X(t)+ D^-_yJ(t,1-L(T-t)))_+ - D^-_yJ(t,1-L(T-t))u^{0,0}(t) \right) dt\right]
\\ &=& (I)+(II)+(III).
\end{eqnarray*}
Note that
\begin{eqnarray*}
 E\left[\int_0^{\bar\sigma}   E[\bar M(\bar \sigma)|\mathcal{F}_t] u^{0,0}(t) dt\right]=\int_0^{T}   E\left[{\bf 1}_{\{\bar\sigma>t\}}\bar M(\bar \sigma) u^{0,0}(t)\right]dt
= E\left[\bar M(\bar \sigma) \int_0^{\bar\sigma} u^{0,0}(t)dt\right].
\end{eqnarray*}
Hence,
$$
(I)=E\left[\int_0^{\bar\sigma}  L(X(t)- \bar M(t))_+ dt\cb{+} \bar M(\sigma) \int_0^{\bar\sigma} u^{0,0}(t)dt \right].
$$
Moreover, by Theorem \ref{thm:smooth2}, we have,
$$
(II)=E\left[{\bf 1}_{\{\bar\sigma_{\upper}\leq \bar\sigma_L\}} \int_{\bar\sigma_{\upper}}^T   Y^*(t)u^{0,0}(t)  dt\right]=0=E\left[{\bf 1}_{\{\bar\sigma_{\upper}\leq \bar\sigma_L\}} \bar M(\bar \sigma) \int_{\bar\sigma_{\upper}}^T u^{0,0}(t)dt\right],
$$
because $u^{0,0}(t)=0$ for $t>\bar \sigma_{\upper}$, and
$$
(III)=E\left[{\bf 1}_{\{\bar\sigma_{\upper}> \bar\sigma_L\}} \int_{\bar\sigma_L}^T   LX(t) dt\right].
$$
Hence,
\begin{eqnarray*}
 \bar J(0,0)&=&E\left[\int_0^{T}  L(X(t)- \bar M(t))_+ dt\right]+E\left[\bar M(\bar \sigma) \int_0^{\bar\sigma} u^{0,0}(t)dt\right] \\ &&
- E\left[{\bf 1}_{\{\bar\sigma_{\upper}\leq\bar\sigma_L\}} \int_{\bar\sigma}^T    L(X(t)- \bar M(t))_+ dt  \right]+E\left[{\bf 1}_{\{\sigma_{\upper}\leq \sigma_L\}} \bar M(\bar \sigma) \int_{\bar\sigma}^T u^{0,0}(t)dt\right]\\&&+
E\left[{\bf 1}_{\{\bar\sigma_{\upper}> \bar\sigma_L\}} \int_{\bar\sigma}^T  \left( LX(t) -  L(X(t)- \bar M(t))_+ \right)dt\right].
\end{eqnarray*}
Now, by the definition of $\bar M$ and the supermartingale property of $Y^*$ we obtain for $t>\bar\sigma$ on $\bar\sigma_{\upper}\leq \bar\sigma_L$,
$$
X(t)\leq Y^*(t)\leq Y^*(\bar \sigma)+M^*(t)-M^*(\bar \sigma)=\bar M(t)
$$
and analogously \cb{for }$t>\bar \sigma$ on $\bar\sigma_L< \bar\sigma_{\upper}$, using the submartingale property of $Y_*$
$$
X(t)\geq Y_*(t)\geq Y_*(\bar \sigma)+M_*(t)-M_*(\bar \sigma)=\bar M(t).
$$
Thus,
\begin{eqnarray*}
 \bar J(0,0)&=&E\left[\int_0^{T}  L(X(t)- \bar M(t))_+ dt\right]+E\left[\bar M(\bar \sigma) \int_0^{\bar\sigma} u^{0,0}(t)dt\right] \\ &&
+E\left[{\bf 1}_{\{\sigma_{\upper}\leq \sigma_L\}} \bar M(\bar \sigma) \int_{\bar\sigma}^T u^{0,0}(t)dt\right]+
E\left[{\bf 1}_{\{\bar\sigma_{\upper}> \bar\sigma_L\}} \int_{\bar\sigma}^T  L\bar M(t)dt\right]
\end{eqnarray*}
Noting that
\begin{eqnarray*}
E\left[{\bf 1}_{\{\bar\sigma_{\upper}> \bar\sigma_L\}} \int_{\bar\sigma}^T  L\bar M(t)dt\right]= E\left[{\bf 1}_{\{\bar\sigma_{\upper}> \bar\sigma_L\}} \bar M(\bar \sigma)\int_{\bar\sigma_L}^T  Ldt\right]
=E\left[{\bf 1}_{\{\bar\sigma_{\upper}> \bar\sigma_L\}} \bar M(\bar \sigma)\int_{\bar\sigma}^T  u^{0,0}(t) dt\right]
\end{eqnarray*}
by (\ref{U'})  and that $\int_0^T u^{0,0}(t)dt=1$, we finally  obtain
$$
\bar J(0,0)=E\left[\int_0^{T}  L(X(t)- \bar M(t))_+ dt\right]+E\left[\bar M(\bar \sigma)\right]=E\left[\int_0^{T}  L(X(t)- \bar M(t))_+ dt\right]+E[\bar M(0)].
$$
\end{proof}

\subsubsection*{Acknowledgement}
The authors gratefully acknowledge financial support by the ATN-DAAD Australia Germany Joint Research Cooperation Scheme.

\end{document}